\documentclass[10pt]{article}

\usepackage{amsmath,amssymb,amsbsy,amsthm,color,epsfig,mathrsfs}
\usepackage{amsbsy,floatrow}
\usepackage{caption}
\usepackage{pifont}

\usepackage{hyperref}
\usepackage{bookmark}

\newtheorem{theorem}{Theorem}

\newtheorem{condition}[theorem]{Condition}

\newtheorem{definition}[theorem]{Definition}

\newtheorem{lemma}[theorem]{Lemma}

\newtheorem{proposition}[theorem]{Proposition}

\newcounter{spslist}

\newcounter{geqncount}
    {\refstepcounter{equation}%
     \setcounter{geqncount}{\value{equation}}%
     \setcounter{equation}{0}%
  }%
    {\setcounter{equation}{\value{geqncount}}}

\newcommand{\Hh}{{\mathcal H}}
\newcommand{\Ss}{{\mathcal S}}

\newcommand{\Cf}{{\mathfrak C}}
\newcommand{\Df}{{\mathfrak D}}
\newcommand{\Gf}{{\mathfrak G}}
\newcommand{\Af}{{\mathfrak A}}
\newcommand{\Bf}{{\mathfrak B}}
\newcommand{\Mf}{{\mathfrak M}}
\newcommand{\CC}{\mathbb{C}}
\newcommand{\ZZ}{\mathbb{Z}}
\newcommand{\RR}{\mathbb{R}}
\newcommand{\TT}{\mathbb{T}}
\newcommand{\V}{{\mathcal V}}
\newcommand{\E}{{\mathcal E}}
\newcommand{\dE}{\vec{\mathcal E}}
\newcommand{\dir}{\vec}

\newcommand{\G}{\Gamma}
\newcommand{\Go}{\mathring{\Gamma}}
\newcommand{\Vo}{\mathring{\mathcal V}}
\newcommand{\Eo}{\mathring{\mathcal E}}
\newcommand{\Ao}{\mathring{A}}
\newcommand{\Wo}{\mathring{W}}
\newcommand{\Afo}{\mathring{\mathfrak A}}

\newcommand{\cc}{\check{c}}
\newcommand{\cs}{\check{s}}
\newcommand{\so}{\mathring{s}}
\newcommand{\co}{\mathring{c}}
\newcommand{\cto}{\tilde{\mathring{c}}}

\newcommand{\op}{_1^+}
\newcommand{\tp}{_2^+}
\newcommand{\om}{_1^-}
\newcommand{\tm}{_2^-}

\newcommand{\Hloc}{H^{2R}_\text{loc}}

\newcommand{\itilde}{\tilde{\imath}}

\newcommand{\dom}{\mathcal{D}}

\newcommand{\tq}{{\tilde q}}

\newcommand{\half}{{\textstyle\frac{1}{2}}}

\newcommand{\Dfcn}{\Omega}

\begin{document}

\bibliographystyle{plain}

\begin{center}
{\bf \Large  Reducible Fermi Surfaces for \\\vspace{0.7ex} Non-symmetric Bilayer Quantum-Graph Operators}
\end{center}

\vspace{0.2ex}

\begin{center}
{\scshape \large Stephen P\hspace{-2pt}. Shipman}\footnote{shipman@lsu.edu; 303 Lockett Hall, Department of Mathematics, Louisiana State University, Baton Rouge, LA 70803} \\
\vspace{1ex}
{\itshape Department of Mathematics, Louisiana State University\\
Baton Rouge, LA 70803, USA}
\end{center}

\vspace{3ex}
\centerline{\parbox{0.9\textwidth}{
{\bf Abstract.}\
This work constructs a class of non-symmetric periodic Schr\"odinger operators on metric graphs (quantum graphs) whose Fermi, or Floquet, surface is reducible.
The Floquet surface at an energy level is an algebraic set that describes all complex wave vectors admissible by the periodic operator at the given energy.
The graphs in this study are obtained by coupling two identical copies of a periodic quantum graph by edges to form a bilayer graph. 
Reducibility of the Floquet surface for all energies ensues when the coupling edges have potentials belonging to the same asymmetry class.
The notion of asymmetry class is defined in this article through the introduction of an entire spectral A-function $a(\lambda)$ associated with a potential---two potentials belong to the same asymmetry class if their A-functions are identical.
Symmetric potentials correspond to $a(\lambda)\equiv0$.
If the potentials of the connecting edges belong to different asymmetry classes, then typically the Floquet surface is not reducible.
An exception occurs when two copies of certain bipartite graphs are coupled; the Floquet surface in this case is always reducible.  This includes AA-stacked bilayer graphene.
}}

\vspace{3ex}
\noindent
\begin{mbox}
\end{mbox}

\vspace{3ex}

\noindent
\begin{mbox}
\end{mbox}
\vspace{3ex}

\hrule

\section{Introduction} 

The Fermi surface for a periodic Schr\"odinger operator $-\nabla^2+q(x)$ ($x\in\RR^n$) is the analytic set of complex wavevectors $(k_1,\dots,k_n)$ for which the operator admits a (non-square-integrable) state at a fixed energy~$\lambda$.  For periodic Schr\"odinger operators on metric graphs, known as quantum graphs, the Fermi surface is an algebraic set in the variables $(z_1,\dots,z_n)=(e^{ik_1},\dots,e^{ik_n})$---that is, the zero set of a polynomial in several variables.  The reducibility of the Fermi surface into the union of two algebraic sets is important because it is intimately related to the existence of embedded eigenvalues induced by a local perturbation of the operator.  It is proved by Kuchment and Vainberg~\cite{KuchmentVainberg2006} that reducibility is required for a local perturbation to engender a square-integrable eigenfunction with unbounded support at an energy that is embedded in the continuous spectrum.

The type of Fermi surface considered in this article is the zero set of a single Laurent polynomial in $(z_1,\dots,z_n)$, and its reducibility is equivalent to the nontrivial factorability of this polynomial into a product of two Laurent polynomials.
That a polynomial in several variables generically cannot be factored nontrivially indicates that a quantum graph must possess special features in order that a local defect be able to support an embedded eigenvalue with eigenfunction having unbounded support.  The typical feature is symmetry.  For a class of operators possessing reflectional symmetry, it is proved in~\cite[\S3]{Shipman2014} that such eigenfunctions are possible due to the decomposition of the operator on even and odd states, which, in turn, effects a canonical reduction of the Fermi surface.
The present work addresses the reducibility of the Fermi surface for a certain class of quantum graphs that are not decomposable by symmetry.  Although the underlying metric graphs are reflectionally symmetric, the Schr\"odinger operators on them are not.  The construction of embedded eigenvalues is not investigated~here.

Reducibility or irreducibility of the Fermi surface has been established for very few periodic operators. 
It is irreducible for all but finitely many energies for the discrete two-dimensional Laplacian plus a periodic potential~\cite{GiesekerKnorrerTrubowitz1993} and for the continuous Laplacian plus a separable potential in two and three dimensions~\cite{BattigKnorrerTrubowitz1991,KuchmentVainberg2000}.  And as mentioned already, it is reducible for quantum graphs that admit a reflectional symmetry.  This work adds to that list a large class of non-reflectionally-symmetric operators with reducible Fermi surface.

The fundamental object of study is the ``bilayer graph" obtained by coupling two identical copies of a given periodic quantum graph by edges connecting corresponding vertices, as shown in Fig.~\ref{fig:BilayerGraphene}.  The potential $q_e(x)$ of the  Schr\"odinger operator $-d^2/dx^2 + q_e(x)$ on a connecting edge $e$ imparts asymmetry to the graph.  An entire spectral function $a(\lambda)$, called the A-function, associated to any potential $q(x)$, is introduced (equation~\ref{Afunction}, \S\ref{sec:spectralfunctions}).  Its significance is that it decides whether two different potentials are compatible with regard to their asymmetry.  One of the main results of this work, Theorem~\ref{thm:asymmetric1} in section~\ref{sec:asymmetric1}, is that {\em reducibility of the Floquet surface occurs when the coupling edges have potentials belonging to the same asymmetry class,} two potentials being in the same class if their spectral A-functions $a(\lambda)$ are identical.  Symmetric potentials correspond to $a(\lambda)\equiv0$.
When two copies of a bipartite graph with two vertices per period are joined, the resulting bilayer graph turns out to have reducible Fermi surface regardless of the two asymmetry classes of the connecting edges, as reported in Theorem~\ref{thm:graphene} in section~\ref{sec:asymmetric2}.  This includes the form of bilayer graphene in which the two layers are exactly aligned; this is commonly known as AA-stacked graphene, as distinguished from AB-stacked graphene.

To construct a bilayer graph, one starts with a given periodic quantum graph $(\Go,\Ao)$, where $\Go$ is a metric graph and $\Ao$ is a Schr\"odinger operator defined on it (section~\ref{sec:bilayer}).  The periodicity means that there is a faithful $\ZZ^n$ action on $\Go$ that commutes with $\Ao$.  The action of $g=(g_1,\dots,g_n)\in\ZZ^n$ is viewed as a shift along the vector $g_1v_1+\dots+g_nv_n$, with $\{v_j\}_{j=1}^n$ being independent generating period vectors.  This quantum graph is declared to be a single layer.  A~bilayer periodic quantum graph $(\G,A)$ is then built by taking two copies of the single layer and coupling them by edges that connect corresponding vertices, as in Fig~\ref{fig:BilayerGraphene}.  Although the bare bilayer metric graph possesses reflectional symmetry about the centers of the connecting edges, the potentials associated to the connecting edges are allowed to be asymmetric.   

\begin{figure}[ht]
\centerline{
\scalebox{0.25}{\includegraphics{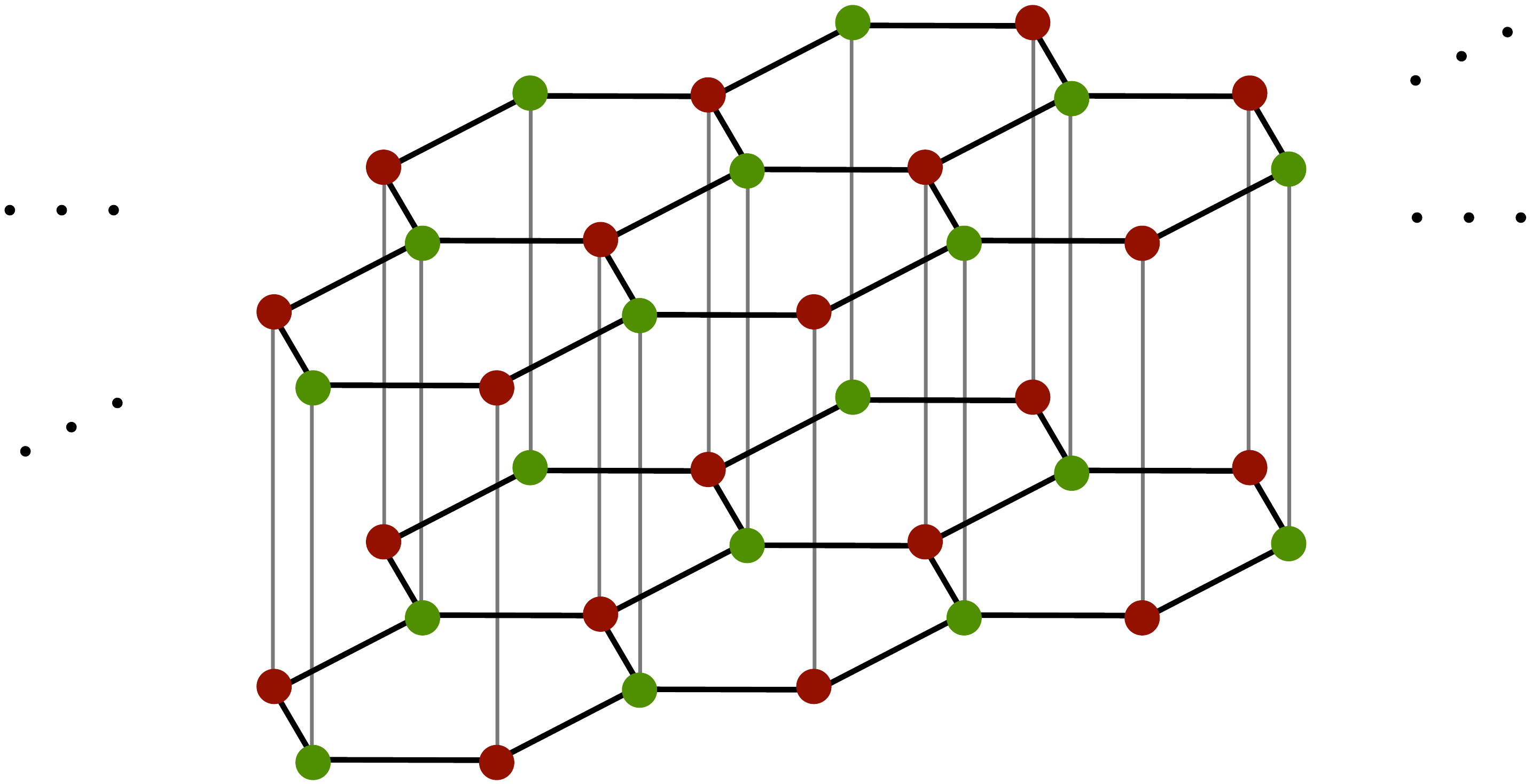}}
}
\caption{\small An example of a periodic bilayer quantum graph based on the single-layer hexagonal graphene structure.}
\label{fig:BilayerGraphene}
\end{figure}

Floquet modes $u(x)$ of $(\G,A)$ are simultaneous eigenfunctions of $A$ and the $\ZZ^n$ action,
\begin{equation}
  Au = \lambda u\,,
  \quad
  u(g\,\cdot) = z^g u
\end{equation}
for $g=(g_1,\dots,g_n)\in\ZZ^n$.  Here, $\lambda\in\CC$ is called the {\em energy} and $z=(z_1,\dots,z_n)\in(\CC^*)^n$ is called the vector of {\em Floquet multipliers} associated with the $n$ fundamental shifts, and $z^g=z_1^{g_1} \cdots z_n^{g_n}$.  The {\em wavevector} $k=(k_1,\dots,k_n)\in\CC^n$ is such that
$z=e^{ik}$, that is $z_\ell=e^{ik_\ell}$ for $\ell=1,\dots,n$, and its components are defined up to the addition of an integer multiple of $2\pi i$.  Floquet modes are never square integrable.

The {\em dispersion relation} $D(\lambda,z)=0$ describes the pairs $(\lambda,z)$ for which $(\G,A)$ admits a Floquet mode.  The function $D(\lambda,z)$ is called the {\em dispersion function}; it is analytic in $\lambda$ (except on a discrete set~(\ref{sigmaD})) and a Laurent polynomial in $z$.  As a subset of $\CC\times(\CC^*)^n$, the dispersion relation is known as the Bloch variety or dispersion surface of $(\G,A)$.  The present study concerns this relation for fixed energy $\lambda$ considered as an algebraic curve in $(\CC^*)^n$.  This set $\Phi_\lambda=\Phi_{A,\lambda}$ is called the {\em Floquet surface} or Floquet variety for $A$ at energy~$\lambda$,
\begin{equation}\label{Floquetsurface}
  \Phi_\lambda := \left\{ z\in(\CC^*)^n : D(\lambda,z)=0 \right\}.
\end{equation}
When considered as a function of the wavevector $k$, it is the Fermi surface.

The reducibility of $\Phi_\lambda$ is equivalent to the factorization of $D(\lambda,z)$ as a Laurent polynomial in the several variables $z=(z_1,\dots,z_n)$ into two distinct nontrivial Laurent polynomial factors, $D(\lambda,z)=D_1(\lambda,z)D_2(\lambda,z)$, in the sense that each factor has a nonempty zero set in $(\CC^*)^n$.  There is no requirement on the nature of the $\lambda$-dependence of the coefficients of the factors.  This reveals a deep connection between factorability of polynomials in several variables and embedded eigenvalues.  The connection transpires through the Floquet transform (the Fourier transform with respect to the $\ZZ^n$ symmetry group) as demonstrated in~\cite{KuchmentVainberg2006} in the proof of Theorem~6 (p.~679).

\smallskip
This article establishes the following hierarchy of symmetry properties of periodic bilayer quantum graphs and the corresponding reducibility properties of the Floquet surface.

\begin{enumerate}

\item\label{item:symmetric} {\slshape Symmetric connecting edges: vanishing A-function $a(\lambda)=0$.} Reducibility of the Floquet surface $\Phi_\lambda$ is obtained by a canonical decomposition of the bilayer graph operator onto the even and odd states with respect to the reflectional symmetry about the connecting edges.
This was shown in~\cite{Shipman2014} and is reviewed in section~\ref{sec:symmetric}.

\item {\slshape Asymmetric connecting edges with identical A-functions $a(\lambda)$.}  Reducibility of the Floquet surface $\Phi_\lambda$ is obtained by an energy-dependent decomposition of the operator reduced to the generalized (not square integrable) eigenspace of Floquet modes for each given energy $\lambda$.
The proof is in~section~\ref{sec:asymmetric1}.

\item\label{item:graphene} {\slshape Asymmetric connecting edges with different A-functions $a(\lambda)$:
bipartite layers with two vertices per period.}  Reducibility of the Floquet surface $\Phi_\lambda$ is obtained due to a special reduction of the dispersion function $D(\lambda,z)$ to a function of a single variable that is a polynomial in $(z_1,\dots,z_n)$.  The analysis is carried out in section~\ref{sec:asymmetric2}.

\item\label{item:general} {\slshape Asymmetric connecting edges with different A-functions $a(\lambda)$: general case.}  Irreducibility of the Floquet surface $\Phi_\lambda$ appears to be generic.
It is proved for a simple example in section~\ref{sec:asymmetric3}.

\end{enumerate}

When the connecting edges are symmetric (case~\ref{item:symmetric}), the components of the bilayer quantum graph $(\G,A)$ acting in the even and odd spaces of states can be realized as operators on ``decorated" copies of the single-layer graph $\Go$.  These decorations are obtained by attaching a dangling edge to each vertex of the single layer and imposing a self-adjoint condition on the terminal, or free, vertex of that edge.
 It turns out that, when the connecting edges are not symmetric, such a realization is not possible (Theorem~\ref{thm:realizability}, section~\ref{sec:realizability}) due to branch points on the Riemann surface of the Dirichlet-to-Neumann map for the connecting edges.

\smallskip
Section~\ref{sec:edge} introduces the spectral A-function and defines the notion of an asymmetry class of connecting edges.
Section~\ref{sec:bilayer} lays down the precise construction of periodic bilayer quantum graphs.
Section~\ref{sec:symmetric} reviews the case of symmetric coupling potentials, and Section~\ref{sec:asymmetric1} proves the main Theorem~\ref{thm:asymmetric1} on reducibility for coupling potentials with compatible asymmetries.  Section~\ref{sec:asymmetric2} presents the special case of bipartite layers, and Section~\ref{sec:asymmetric3} addresses generic periodic bilayer quantum graphs.

\section{Analysis of a single edge and asymmetry}\label{sec:edge} 

This section introduces the spectral asymmetry function, or A-function, associated to a potential $q$ on a finite interval or an edge of a graph and describes some of its properties.

Let the edge $e$ connecting the vertices $v$ and $w$ be parameterized by identifying it with the interval $[0,1]$, directed from $v$ to $w$.  This allows one to define a Schr\"odinger operator $-d^2/dx^2+q(x)$ on $e$, with $q\in L^2([0,1],\RR)\cong L^2(e,\RR)$.  
When the same edge $e$ is considered as directed from $w$ to $v$, it is identified with $[0,1]$ by replacing $x$ with $1\!-\!x$, and the Schr\"odinger operator now takes the form $-d^2/dx^2+\tq(x)$, in which $\tq(x)$ denotes the reflection of~$q$ about the center of the interval $[0,1]$,
\begin{equation}
  \tq(x) = q(1-x)\,.
\end{equation}
In order to avoid burdensome technicalities, it is sometimes convenient to identify a function $u:e\to\CC$ with the corresponding function of the parameter $x$ by writing $u:[0,1]\to\CC$.  Let a function $u$ in the Sobolev space $H^2(0,1)\cong H^2(e)$ of square-integrable functions with square-integral derivatives satisfy the eigenvalue condition
\begin{equation}
  (-d^2/dx^2+q(x)-\lambda)u = 0
\end{equation}
on $(0,1)$.  Since the vertices $v$ and $w$ are identified with the points $0$ and $1$, one can write
\begin{equation}\label{vertexdata}
\begin{split}
  u(v)=u(0), &\quad u(w)=u(1) \\
  u'(v)=\frac{du}{dx}(0), &\quad u'(w)=-\frac{du}{dx}(1).
\end{split}
\end{equation}
The reason for the minus sign is that it is appropriate to treat the two vertices symmetrically, taking the derivative at each vertex in the direction pointing from the vertex into the edge.

\subsection{Spectral functions for an edge: The A-function}\label{sec:spectralfunctions}

Given an edge $e$ directed from $v$ to $w$ with parameter $x\in[0,1]$,
let $c_q(x,\lambda)$ and $s_q(x,\lambda)$ be a fundamental pair of solutions to $(-d^2/dx^2+q(x)-\lambda)u=0$ satisfying the initial conditions
\begin{equation}
\begin{split}
  c_q(0,\lambda) = 1, &\quad s_q(0,\lambda) = 0 \\
  c'_q(0,\lambda) = 0, &\quad s'_q(0,\lambda) = 1,
\end{split}
\end{equation}
in which the prime denotes the derivative with respect to the first argument $x$.
Define
\begin{equation}
\begin{split}
  c_{(v,w)}(\lambda) = c(\lambda) := c_q(1,\lambda), &\qquad s_{(v,w)}(\lambda) = s(\lambda) := s_q(1,\lambda), \\
  c'_{(v,w)}(\lambda) = c'(\lambda) := c'_q(1,\lambda), &\qquad s'_{(v,w)}(\lambda) = s'(\lambda) := s'_q(1,\lambda),
\end{split}
\end{equation}
in which the dependence on $q$ is suppressed.  When the edge $e$ is directed from $w$ to $v$, the corresponding quantities have $\tq$ in place of $q$; they will also be abbreviated by use of a tilde.  For example,
\begin{equation}
  c_{(w,v)}(\lambda) = \tilde c(\lambda) := c_{\tq}(1,\lambda).
\end{equation}
These four functions are entire and of exponential order $1/2$ and have their roots on the real line~\cite[\S1.1]{FreilingYurko2001}.
An important observation is that $s(\lambda)$ is a function of the undirected edge $\left\{ v,w \right\}$, whereas $c(\lambda)$ is a function of the directed edge $(v,w)$.

There are relations between these spectral functions for $q(x)$ and $\tilde q(x)$.
Let $u(x)$ satisfy $-u'' + q(x)u=\lambda u$; then
\begin{equation}\label{transferq1}
\renewcommand{\arraystretch}{1.2}
\left[\hspace{-5pt}
\begin{array}{cc}
  c(\lambda) & s(\lambda) \\ -c'(\lambda) & -s'(\lambda)
\end{array}
\hspace{-3pt}\right]
\renewcommand{\arraystretch}{1.2}
\left[\hspace{-5pt}
\begin{array}{c}
  u(0) \\ u'(0)
\end{array}
\hspace{-5pt}\right]
=
\renewcommand{\arraystretch}{1.2}
\left[\hspace{-5pt}
\begin{array}{c}
  u(1) \\ -u'(1)
\end{array}
\hspace{-5pt}\right].
\end{equation}
The function $\tilde u(x)=u(1-x)$ satisfies $-\tilde u''+\tilde q(x)\tilde u=\lambda \tilde u$; therefore
\begin{equation*}
\renewcommand{\arraystretch}{1.2}
\left[\hspace{-5pt}
\begin{array}{cc}
  \tilde c(\lambda) & \tilde s(\lambda) \\ -\tilde c'(\lambda) & -\tilde s'(\lambda)
\end{array}
\hspace{-3pt}\right]
\renewcommand{\arraystretch}{1.2}
\left[\hspace{-5pt}
\begin{array}{c}
  \tilde u(0) \\ \tilde u'(0)
\end{array}
\hspace{-5pt}\right]
=
\renewcommand{\arraystretch}{1.2}
\left[\hspace{-5pt}
\begin{array}{c}
  \tilde u(1) \\ -\tilde u'(1)
\end{array}
\hspace{-5pt}\right],
\end{equation*}
or, equivalently,
\begin{equation}\label{transferq2}
\renewcommand{\arraystretch}{1.2}
\left[\hspace{-5pt}
\begin{array}{cc}
  \tilde c(\lambda) & \tilde s(\lambda) \\ -\tilde c'(\lambda) & -\tilde s'(\lambda)
\end{array}
\hspace{-3pt}\right]
\renewcommand{\arraystretch}{1.2}
\left[\hspace{-5pt}
\begin{array}{c}
  u(1) \\ -u'(1)
\end{array}
\hspace{-5pt}\right]
=
\renewcommand{\arraystretch}{1.2}
\left[\hspace{-5pt}
\begin{array}{c}
  u(0) \\ u'(0)
\end{array}
\hspace{-5pt}\right].
\end{equation}
Thus the matrices in (\ref{transferq1}) and (\ref{transferq2}) are inverses of each other.  They also have determinant equal to $-1$.
This yields the relations $\tilde c=s'$, $s=\tilde s$, and $c'=\tilde c'$.
From these, one obtains the following relations.
The first two demonstrate that $s(\lambda)$ and $c'(\lambda)$ do not depend on the direction of the edge.  The notation with the subscript $\{v,w\}$ emphasizes that these quantities are determined by the undirected edge alone,
\begin{equation}\label{csrelations}
\begin{split}
  s_{\{v,w\}} :=\;\;& s_{(v,w)} = s = \tilde s = s_{(w,v)}\,, \\
  c'_{\{v,w\}} :=\;\;& c'_{(v,w)} = c' = \tilde c' = c'_{(w,v)}\,, \\
  & c_{(v,w)} = c = \tilde s' = s'_{(w,v)}\,, \\
  & s'_{(v,w)} = s' = \tilde c = c_{(w,v)}\,.
\end{split}
\end{equation}

Define the {\em transfer matrix} $T_q(\lambda)$ for the edge $e$ directed from $v$ to $w$, for the potential $q$ and spectral value~$\lambda$, to be the matrix that takes Cauchy data $(u(v),u'(v))$ at $v$ to Cauchy data $(u(w),u'(w))$ at $w$.
Define the {\em Dirichlet-to-Neumann}, or DtN, matrix $G_q(\lambda)$ as that which takes Dirichlet data $(u(v),u(w))$ to Neumann data $(u'(v),u'(w))$.  Using the identity $s'(\lambda)=\tilde c(\lambda)$, one obtains
\begin{align}
\label{Tq}
\underbrace{
\renewcommand{\arraystretch}{1.2}
\left[\hspace{-5pt}
\begin{array}{cc}
  c(\lambda) & s(\lambda) \\ -c'(\lambda) & -\tilde c(\lambda)
\end{array}
\hspace{-3pt}\right]
}_{T_q(\lambda)}
\renewcommand{\arraystretch}{1.2}
\left[\hspace{-5pt}
\begin{array}{c}
  u(v) \\ u'(v)
\end{array}
\hspace{-5pt}\right]
& =
\renewcommand{\arraystretch}{1.2}
\left[\hspace{-5pt}
\begin{array}{c}
  u(w) \\ u'(w)
\end{array}
\hspace{-5pt}\right]
 \\
 \label{Gq}
 \underbrace{
\renewcommand{\arraystretch}{1.2}
\frac{1}{s(\lambda)}
\left[\hspace{-5pt}
\begin{array}{cc}
  -c(\lambda) & 1 \\ 1 & -\tilde c(\lambda)
\end{array}
\hspace{-3pt}\right]
}_{G_q(\lambda)}
\renewcommand{\arraystretch}{1.2}
\left[\hspace{-5pt}
\begin{array}{c}
  u(v) \\ u(w)
\end{array}
\hspace{-5pt}\right]
& =
\renewcommand{\arraystretch}{1.2}
\left[\hspace{-5pt}
\begin{array}{c}
  u'(v) \\ u'(w)
\end{array}
\hspace{-5pt}\right].
\end{align}
$G_q(\lambda)$ is a meromorphic function with simple poles at the roots of $s(\lambda)$, which are the Dirichlet eigenvalues of $-d^2/dx^2+q(x)$ on~$e$~\cite[Lemma~1.1.1]{FreilingYurko2001}.  It is akin to the Weyl-Titchmarsh M-function for an interval~\cite[\S1.4.4]{FreilingYurko2001}.

The potential $q$ is uniquely decomposed into symmetric and anti-symmetric parts with respect to reflection about the center of $e$,
\begin{align}
  q(x) &= q_+(x) + q_-(x),  \\
  \tq(x) &= q_+(x) - q_-(x).
\end{align}
Define two entire spectral functions associated with the potential $q$,
\begin{align}
  a_q(\lambda) = a(\lambda) &= \half\big( c(\lambda) - \tilde c(\lambda) \big),\label{Afunction}
  \qquad \text{(spectral A-function)} \\
  b_q(\lambda) = b(\lambda) &= \half\big( c(\lambda) + \tilde c(\lambda) \big). \label{Bfunction}
\end{align}
The first of these shall be known as the {\em spectral asymmetry function}, or ``A-function" associated with $q(x)$.
The A-function can is half the trace of the transfer matrix,
\begin{equation}
  a_q(\lambda) \;=\; \half \,\mathrm{tr}\, T_q(\lambda),
\end{equation}
and because of the relations~(\ref{csrelations}), it can be written in other ways, such as
\begin{equation}\label{ausings}
  a(\lambda)=\half\big(\tilde s'(\lambda)-s'(\lambda)\big).
\end{equation}

This A-function and the function $b(\lambda)$ are identical to the functions $u_-(\lambda)$ and $u_+(\lambda)$ defined in~\cite{MarcenkoOstrovskii1975} (see p.~494 and Lemma~4.1), in which the authors give a characterization of the all spectra of Hill operators ($-d^2/dx^2+q(x)$ with periodic potential $q$ on $\RR$) as certain sequences of intervals on the real $\lambda$-line.  The A-function is also identical to $\delta(\lambda)$ defined in~\cite[p.~2]{Yurko2016}, and $b(\lambda)$ is equal to $\Delta(\lambda)$ there.  In~\cite{Yurko2016}, the authors are interested in symmetric potentials, for which $a(\lambda)$ vanishes identically.

If $q$ is considered to be a function defined on an edge $e=\{v,w\}$ rather than a function of $x\in[0,1]$, a direction must be specified in order to determine the sign of the A-function; precisely,
\begin{equation}
  a_{(v,w),q}(\lambda) \;=\; -a_{(w,v),q}(\lambda) \;=\; \half \big( c_{(v,w)}(\lambda)-c_{(w,v)}(\lambda) \big) \,.
\end{equation}

\begin{definition}
Two potentials $q_1$ and $q_2$ in $L^2([0,1],\RR)$ are said to be in the same {\em asymmetry class} if their associated A-functions  are identical, that is, $a_{q_1}(\lambda)=a_{q_2}(\lambda)$ for all $\lambda\in\CC$.  As potentials on an edge $e=\{v,w\}$, $q_1$ and $q_2$ are in the same asymmetry class if their A-functions associated with a given direction of the edge are identical, that is,
$a_{(v,w),q_1}(\lambda)=a_{(v,w),q_2}(\lambda)$.
Potentials in the same asymmetry class are said to have {\em compatible asymmetries}.
\end{definition}

That this is a good definition is manifest by the following theorem, which relies on a uniqueness theorem by G.\,Borg for an inverse spectral problem. 
The first part of the theorem implies that the symmetric potentials $q(x)=q(1-x)$ form a single asymmetry class associated with $a(\lambda)=0$.  Characterizing the asymmetry classes for nonzero $a(\lambda)$ is a difficult non-unique inverse problem that deserves an investigation of its own.

\begin{theorem}[Properties of the A-function]\label{thm:a}   
The A-function for square-integrable potentials satisfies the following properties.
\begin{enumerate}
 \item
The potential $q(x)$ is symmetric if and only if $a(\lambda)$ vanishes identically, that is,
\begin{equation}\label{equivalence}
  q(x) = q(1-x)
  \quad\iff\quad
  a(\lambda) = 0\,.
\end{equation}
The first equality is in the sense of $L^2$ (almost every $x$), and the second means for all $\lambda\in\CC$.
\item
If $\lambda\in\RR$, then $a(\lambda)\in\RR$.
\item
For $i\in\{1,2\}$, let $q_i(x)\in L^2[0,1]$ have A-function $a_i(\lambda)$ and DtN matrix $G_i(\lambda)$; and let $\lambda$ not be a Dirichlet eigenvalue for $-d^2/dx^2 + q_i(x)$ ($i\in\{1,2\}$).  Then $a_1(\lambda)=a_2(\lambda)$ if and only if $G_1(\lambda)$ and $G_2(\lambda)$ commute with each other.
In this case, if $\lambda\in\RR$, then $G_1(\lambda)$ and $G_2(\lambda)$ are simultaneously diagonalizable.
\item The Dirichlet spectrum of $-d^2/dx^2+q(x)$ on an edge \,$e$\, together with the A-function of \,$q$\, determine the potential \,$q$\, uniquely.
\item
The A-function $a(\lambda)$ associated with the potential $q\in L^2([0,1],\RR)$ satisfies
\begin{equation}\label{Intqcc}
  c'(\lambda)\,a(\lambda) \,=\, -\int_0^1 q_-(x)\,c(x,\lambda)\, \tilde c(x,\lambda)\,dx\,,
\end{equation}
 in which $q_-(x)=\half\left( q(x)-q(1-x) \right)$ is the odd part of $q(x)$.
\item
For all $\lambda_0\in\CC$, if $a(\lambda_0)^2+1=0$, then $G(\lambda_0)$ has a double eigenvalue of geometric multiplicity $1$ and
\begin{equation}
  \frac{da}{d\lambda}(\lambda_0) \;=\; -\frac{s(\lambda)}{2} \int_0^1 \psi(x)^2\,dx,
\end{equation}
in which $\psi$ is the solution to $-\psi''+(q(x)-\lambda_0)\psi=0$ associated with the one-dimensional eigenspace of $G(\lambda_0)$ with $\psi(0)=1$.
\end{enumerate}
\end{theorem}

\begin{proof}  
A very brief proof of part (1) can be found in~\cite[Lemma~4]{Yurko1975}; a more detailed proof is given here.
It is consequence of Borg's theorem on the determination of $q$ from the spectra for two different boundary conditions (\cite{Borg1946}, \cite[Theorem~1.4.4]{FreilingYurko2001}), which was observed through correspondence~\cite{BrownSchmidtWood2018}.  The argument goes as follows.
The zero set of $c(\lambda)$ is equal to the spectrum for the potential $q$ with boundary conditions $u'(0)=0$ and $u(1)=0$ (N-D spectrum), and the zero set of $c'(\lambda)$ is the spectrum for $q$ with conditions $u'(0)=0$ and $u'(1)=0$ (N-N spectrum).  Given that $a(\lambda)=0$, one has $c(\lambda)=\tilde c(\lambda)$; and due to the relations~(\ref{csrelations}), one also has $c'(\lambda)=\tilde c'(\lambda)$.  This implies that the N-D spectra for both potentials $q$ and $\tilde q$ are identical and the N-N spectra for both potentials are identical.  This is sufficient, by Borg's Theorem, to guarantee the (almost-everywhere) equality of $q$ and $\tilde q$.  

Part (2) results from $c(\lambda)$ and $\tilde c(\lambda)$ being real whenever $\lambda$ is real, which is a consequence of $q$ being real valued.

Part (3) is a straightforward calculation.  Verification of the last sentence uses part (2) and the fact that $s(\lambda)G(\lambda)$ is diagonalizable when $a(\lambda)^2+1\not=0$.

Part (4) ensues from the unique determination of \,$q$\, from its Dirichlet spectrum $\{\mu_n(q)\}_{n=1}^\infty$ and spectral data $\{\kappa_n(q)\}_{n=1}^\infty$ defined in~\cite[p.~59]{PoschelTrubowitz1987}
\begin{equation}
  \kappa_n(q) \;=\;
  \log \left| s'_q(\mu_n(q)) \right| \;=\;
  \sinh^{-1}\!\big( (-1)^n\,a_q(\mu_n(q)) \big),
\end{equation}
in which the last equality uses the expression (\ref{ausings}) and the fact that the $n^\mathrm{th}$ Dirichlet eigenfunction of $-d^2/dx^2+q(x)$ has $n-1$ roots inside the interval $(0,1)$~\cite[Theorem~6,\,p.~41]{PoschelTrubowitz1987}.
This result on unique determination of $q$ is Theorem~5 of~\cite[p.~62]{PoschelTrubowitz1987}.

To prove part (5),  Let $q(x)=q_+(x)+q_-(x)$ be such that $q_+(x)=q_+(1-x)$ and $q_-(x)=-q_-(1-x)$.  The functions $c=c(x,\lambda)$ and $\tilde c = \tilde c(x,\lambda)$ satisfy
\begin{align}
  -c'' + \big(q_+(x)+q_-(x)\big)c - \lambda c &= 0  \\
  -\tilde c'' + \big(q_+(x)-q_-(x)\big)\tilde c - \lambda\tilde c &= 0\,,
\end{align}
in which the prime denotes differentiation with respect to the first argument $x$, with
\begin{equation}
  c(0,\lambda)=\tilde c(0,\lambda)= 1,
  \quad
    c'(0,\lambda)=\tilde c'(0,\lambda)= 0.
\end{equation}
Multiplying the first differential equation by $\tilde c$ and the second by $c$ and then subtracting yields
\begin{equation}
  \left( c\tilde c' - \tilde c c' \right)' + 2q_- c\tilde c \,=\, 0\,.
\end{equation}
Using the initial conditions for $c$ and $\tilde c$ and the fact that $c'(1,\lambda)=\tilde c'(1,\lambda)$, one obtains 
the formula in the proposition.

The proof of part (6) is deferred to the end of section~\ref{sec:riemannsurface}.
\end{proof}

\subsection{Riemann surface for an edge}\label{sec:riemannsurface}  

Analysis of the decomposition of coupled quantum-graph operators and the reduction of the Floquet surface in subsequent sections
is based on the spectral resolution of the the DtN map $G(\lambda)$ for each directed connecting edge $e=(v,w)$.  This spectral resolution is naturally defined on the Riemann surface associated with the characteristic polynomial of $G(\lambda)$.
A complete theory of the spectral resolution of meromorphic operators in finite dimensions on Riemann surfaces is available in~\cite[Ch.3\,\S4]{Baumgartel1985}.

It is convenient to deal with the entire matrix function
\begin{equation}
  s(\lambda) G(\lambda) \,=\,
  -\renewcommand{\arraystretch}{1.2}
\left[
\begin{array}{cc}
  \!\!b(\lambda) & 0 \\
  0 & b(\lambda)\!\!
\end{array}
\right]
  +\renewcommand{\arraystretch}{1.2}
\left[
\begin{array}{cc}
  \!\!-a(\lambda) & 1 \\
  1 & a(\lambda)\!\!
\end{array}
\right]
\end{equation}
and treat the spectral theory of the trace-free part of $s(\lambda) G(\lambda)$,
\begin{equation}
\renewcommand{\arraystretch}{1.2}
  N(\lambda) =
  \left[
\begin{array}{cc}
  \!\!-a(\lambda) & 1 \\
  1 & a(\lambda)\!\!
\end{array}
\right],
\end{equation}
which involves only the A-function of the edge.  The characteristic polynomial in $\mu$ of $N(\lambda)$ is
\begin{equation}
  p(\lambda,\mu) \;=\; \mu^2 - (a(\lambda)^2 + 1).
\end{equation}
One computes that the projection associated to an eigenvalue $\mu$ is
\begin{equation}
  P_{\!\mu} \,=\, \frac{1}{2\mu}
  \renewcommand{\arraystretch}{1.3}
\left[
\begin{array}{cc}
  (a(\lambda)+\mu)^{-1} & 1 \\
  1 & a(\lambda)+\mu
\end{array}
\right].
\end{equation}
This is a meromorphic matrix function on a Riemann surface, described next.

$N(\lambda)$ has an analytic eigenvalue $\mu$ on the Riemann surface defined by the zero-set of $p(\lambda,\mu)$,
\begin{equation}
  \Ss \,=\, \left\{ (\lambda,\mu)\in\CC^2 : \mu^2 = a(\lambda)^2 + 1 \right\}.
\end{equation}
The corresponding eigenvalue of $s(\lambda)G(\lambda)$ is $-b(\lambda)+\mu$, and the other eigenvalue is $-b(\lambda)-\mu$.
The projection $\Ss\to\CC :: (\lambda,\mu)\mapsto\lambda$ is ramified over the set of points
\begin{equation}
  \big\{ \lambda\in\CC : a(\lambda) \in \left\{ i, -i \right\} \big\}
  \qquad
  \text{(ramification points)}.
\end{equation}

About a point $(\lambda_0,\mu)$ with $\mu\not=0$, that is, where $\lambda_0$ is not a ramification point, the variable $\lambda$ serves locally as a complex coordinate for $\Ss$.
At a point $(\lambda_0,0)$ on $\Ss$ above a ramification point $\lambda_0$, the relation $p(\lambda,\mu)=0$ can be written as
\begin{equation}
  \mu^2 = (\lambda-\lambda_0)^n f(\lambda),
  \qquad f(\lambda_0) \not=0, \quad n\geq1.
\end{equation}
If $n=1$, then $\mu$ serves as a local analytic coordinate for $\Ss$ about $(\lambda_0,0)$.
When $n\geq2$, the point $(\lambda_0,0)\in\Ss$ is singular but can be regularized as follows.
Let $\tilde f(\lambda)$ be analytic and non-vanishing in a neighborhood of $\lambda=\lambda_0$ with the property that $\tilde f(\lambda)^2 = f(\lambda)$.
In the case that $n=2m$ is even, there are two sheets above a neighborhood of $(\lambda_0,0)$,
\begin{equation}
  \left\{ (\lambda,\mu) : \mu = (\lambda-\lambda_0)^m \tilde f(\lambda) \right\}
  \quad\text{and}\quad
  \left\{ (\lambda,\mu) : \mu = - (\lambda-\lambda_0)^m \tilde f(\lambda) \right\}.
\end{equation}
In a neighborhood of $\lambda_0$, each sheet has $\lambda$ as an analytic coordinate and the two sheets intersect only in $(\lambda_0,0)$.  By a mild abuse of notation, one can replace the two local sheets by their disjoint union so that $\Ss$ becomes regular at $(\lambda_0,0)$.
In the case that $n=2m+1$ is odd, a neighborhood of $(\lambda_0,0)$ in $\Ss$ can be realized as a connected complex surface by taking $w=\sqrt{\lambda-\lambda_0\,}$ as an analytic coordinate.  More precisely, a neighborhood of $0$ in the $w$-plane maps onto a neighborhood of $(\lambda_0,0)$ in $\Ss$ by means of the map
\begin{equation}
  w \mapsto \big( w^2 + \lambda_0,\, w^{2m+1}\tilde f(w^2 + \lambda_0) \big)\,.
\end{equation}
The projections $P_{\!\mu}$ are meromorphic functions on $\Ss$, regularized as described.  The principal part at a point where $a(\lambda_0)=\pm i$\, is
\begin{equation}
  \frac{\pm1}{2\tilde f(\lambda_0) (\lambda-\lambda_0)^m}
  \renewcommand{\arraystretch}{1.2}
\left[
  \begin{array}{cc}
    \mp i & 1 \\
    1 & \pm i
  \end{array}
\right]
\quad
\text{for $n=2m$}
\end{equation}
and
\begin{equation}
  \frac{1}{2\tilde f(\lambda_0) w^{2m+1}}
  \renewcommand{\arraystretch}{1.2}
\left[
  \begin{array}{cc}
    \mp i & 1 \\
    1 & \pm i
  \end{array}
\right]
\quad
\text{for $n=2m+1$}.
\end{equation}

\begin{proof}[Proof of part (6) of Theorem~\ref{thm:a}]
This proof follows a standard technique for finding derivatives of spectral functions for Schr\"odinger operators in one dimension.
Define
\begin{equation}
  \cc(x,\lambda) = \tilde c(1-x,\lambda),
  \qquad
  \cs(x,\lambda) = \tilde s(1-x,\lambda),
\end{equation}
which satisfy the equation $-u''+ q(x)u=\lambda u$ and the initial conditions
$\cc(1,\lambda)=1$, $d\cc(x,\lambda)/dx|_{x=1}=0$
and $\cs(1,\lambda)=0$, $d\cs(x,\lambda)/dx|_{x=1}=-1$.

Consider the Wronskian
\begin{equation}
  W(x;\lambda_1,\lambda_2) :=
   \renewcommand{\arraystretch}{1.3}
\left|
  \begin{array}{cc}
    c(x,\lambda_1) & \cs(x,\lambda_2) \\
    c'(x,\lambda_1) & \cs'(x,\lambda_2) \\
  \end{array}
\right|,
\end{equation}
whose values at the endpoints are
\begin{align}
  W(0;\lambda_1,\lambda_2) &= \cs'(0,\lambda_2) = -c(1,\lambda_2)  \\
  W(1;\lambda_1,\lambda_2) &= -c(1,\lambda_1).
\end{align}
Using the equation $-u''+ q(x)u=\lambda u$ for both $c$ and $\cs$ yields
\begin{equation}
\begin{split}
  \frac{d}{dx} W(x,\lambda_1,\lambda_2)
  &\;=\;
     \renewcommand{\arraystretch}{1.3}
\left|
  \begin{array}{cc}
    c(x,\lambda_1) & \cs(x,\lambda_2) \\
    c''(x,\lambda_1) & \cs''(x,\lambda_2) \\
  \end{array}
\right| \\
 &\;=\;  - \left|
  \begin{array}{cc}
    c(x,\lambda_1) & \cs(x,\lambda_2) \\
    \lambda_1c(x,\lambda_1) & \lambda_2\cs(x,\lambda_2) \\
  \end{array}
\right| 
\;=\;
-(\lambda_2-\lambda_1) c(x,\lambda_1)\cs(x,\lambda_2).
\end{split}
\end{equation}
Integrating over the $x$-interval $[0,1]$ gives the difference quotient
\begin{equation}
  \frac{c(1,\lambda_1)-c(1,\lambda_2)}{\lambda_1-\lambda_2}
  \;=\;
  -\int_0^1 c(x,\lambda_1)\cs(x,\lambda_2)dx\,,
\end{equation}
which leads to the derivative
\begin{equation}
  \frac{d}{d\lambda}c(1,\lambda) \;=\; -\int_0^1 c(x,\lambda)\cs(x,\lambda)dx\,.
\end{equation}
A similar calculation yields
\begin{equation}
  \frac{d}{d\lambda}\cc(0,\lambda) \;=\; - \int_0^1 \cc(x,\lambda)s(x,\lambda)dx\,.
\end{equation}
The derivative of the A-function is
\begin{equation}\label{Aprime}
  2\frac{d}{d\lambda}a(\lambda)
  \;=\; \frac{d}{d\lambda} \big( c(1,\lambda) - \cc(0,\lambda) \big)
  \;=\; \int_0^1 \big( \cc(x,\lambda)s(x,\lambda) - c(x,\lambda)\cs(x,\lambda) \big)dx\,.
\end{equation}

Now suppose that at some $\lambda\in\CC$, $a(\lambda)^2+1=0$, and let $a(\lambda)=i$ (the case $a(\lambda)=-i$ is treated similarly).
The normal form for $G(\lambda)$ is a single Jordan block with eigenvalue $-b(\lambda)/s(\lambda)$.
The eigenspace is spanned by $v_1=[1,i]^t$.  This means that there is a function $\psi(x)$ satisfying
$-\psi'' + q(x)\psi = \lambda\psi$ and having boundary values
\begin{equation}
  \psi(0)=1,\,\psi'(0)=-b(\lambda)/s(\lambda),
  \qquad
  \psi(1)=i,\,\psi'(1)=ib(\lambda)/s(\lambda).
\end{equation}
A generalized eigenvector is $v_2=[0,1]^t$, meaning that $G(\lambda)v_2=(v_1-b(\lambda)v_2)/s(\lambda)$.  This vector corresponds to a function $\phi(x)$ satisfying
$-\phi'' + q(x)\phi = \lambda\phi$ and having boundary values
\begin{equation}
  \phi(0)=0,\,\phi'(0)=1/s(\lambda),
  \qquad
  \phi(1)=1,\,\phi'(1)=(b(\lambda)-i)/s(\lambda).
\end{equation}
One can verify the following equalities by checking the initial conditions,
\begin{align}
  c(x,\lambda) &= \psi(x) + b(\lambda) \phi(x) \\
  \cc(x,\lambda) &=  (b(\lambda)-i)\psi(x) -ib(\lambda) \phi(x)\\
  s(x,\lambda) &= s(\lambda) \phi(x) \\
  \cs(x,\lambda) &= s(\lambda) \psi(x) -is(\lambda) \phi(x),
\end{align}
from these relations follows
\begin{equation}
  \cc(x,\lambda)s(x,\lambda) - c(x,\lambda)\cs(x,\lambda)
  \;=\;
  -s(\lambda)\,\psi(x)^2
\end{equation}
and thence
\begin{equation}
\begin{split}
  \frac{d}{d\lambda}a(\lambda)
  &\;=\; -\frac{s(\lambda)}{2} \int_0^1 \psi(x)^2 dx\,.
\end{split}
\end{equation}
\end{proof}

\section{Bilayer quantum graphs: Coupling by edges}\label{sec:bilayer}

This section defines a periodic quantum graph and the procedure of coupling two identical graphs with auxiliary edges to form a new periodic quantum graph called a {\em bilayer quantum graph}, formalized in Definition~\ref{def:bilayer}.
Some background and notation is needed for providing a precise construction to support this definition.

This section also describes the energy-dependent reduction of a quantum graph to a combinatorial graph and the Floquet transform and the Floquet surface (or Fermi surface) for these periodic graphs.  The description of quantum graphs essentially follows \cite{BerkolaikoKuchment2013}, but the notation is developed to suit periodic bilayer quantum graphs.

\subsection{Periodic quantum graphs}\label{sec:periodicqg}

A periodic quantum graph $\G$ consists of the following structure.  Some of the notation may seem technical, but it is all very natural.

\smallskip
(1) An underlying graph with vertex set $\V=\V(\G)$ and edge set $\E=\E(\G)$ is endowed with an action by the group $\ZZ^n$ that preserves the vertex-edge incidence and such that $\G/\ZZ^n$ is a finite graph.  The action of $g\in\ZZ^n$ on a vertex or edge of $\G$ is denoted by $v\mapsto gv$ or~$e\mapsto ge$.
A fundamental domain of the $\ZZ^n$ action is denoted by~$W$, and it has by assumption finitely many vertices and edges.

An edge in $e\in\E$ is an unordered set $e=\{v,w\}$ of vertices.  It will be necessary to allow any edge $\{v,w\}$ to assume either of the two directions associated with the ordered pairs $(v,w)$ and~$(w,v)$.  For each vertex $v\in\V$, let $\dE(v)$ denote the set of directed edges incident to $v$, directed away from $v$,
\begin{equation}
  \dE(v) := \{ (v,w) : w\in\V,\, \{v,w\}\in\E \}.
\end{equation}
Thus, if $e=(v,w)\in\dE(v)$, then $\tilde e=(w,v)\in\dE(w)$.  The symbol $e$ may denote either an edge or a directed edge; a directed edge may also be denoted with an arrow $\dir{e}$ when the distinction between undirected and directed is necessary.

\smallskip
(2) $\G$ becomes a metric graph by associating each edge $e=\{v,w\}\in\E$ with an interval $[0,L_e]$.  The directed edge $\dir{e}=(v,w)$ is referred to coordinate $x_{\dir{e}}\in[0,L_e]$ with $x_{\dir{e}}=0$ corresponding to $v$ and $x_{\dir{e}}=L_e$ corresponding to $w$.  The oppositely directed edge $\tilde{\dir{e}}=(w,v)$ is referred to coordinate $x_{\tilde{\dir{e}}}\in[0,L_e]$, with $x_{\dir{e}}+x_{\tilde{\dir{e}}}=L_e$.  Assume that this metric structure is invariant under the action of $\ZZ^n$, and let the action of $g\in\ZZ^n$ on a point $x$ in $\Gamma$ be denoted by~$x\mapsto gx$ ($x$ may be in the interior of an edge or at an endpoint corresponding to a vertex).
This ``metrization" allows one to define standard function spaces on any edge~$e$, such as the Sobolev spaces $H^s(e)$.  

\smallskip
(3) One renders $\G$ a periodic quantum graph by pairing it with a Schr\"odinger operator $A$ that commutes with the $\ZZ^n$ action.  On each edge $e$, $A$ acts by $-D^2 + q_e(x)$.  Here, $D^2=d^2/dx^2_{\dir{e}}$\,, with $\dir{e}$ referring to either direction; and $x$ is any point along $e$.
$A$ acts on functions $f=\{f_e\}_{e\in\E}$ defined on all of $\G$ and that satisfy a Robin condition at each vertex
\begin{equation}\label{Robin} 
  \sum_{e\in\dE(v)} f'_e(v) \,=\, \alpha_v\,f(v)\,.
\end{equation}
Here, $f_e$ is the restriction of a function $f$ on $\G$ to $e$, and if $e=(v,w)$, $f'_e(v)$ is the derivative of $f_e$ at the vertex $v$ directed away from $v$ toward $w$, that is, the derivative from the right of $f_e$ with respect to the coordinate $x_e$ at $x_e=0$.
The vertex condition~(\ref{Robin}) is also known as a $\delta$-type coupling or matching condition~\cite{Exner1997}.

To define $A$ precisely, first set
\begin{equation}
  H^2(\G) \,=\,
  \left\{ f=\{f_e\}_{e\in\E} : f \text{ is continuous; } f_e\in H^2(e) \,\forall e\in\E;\, f,f',f''\in L^2(\G) \right\},
\end{equation}
in which derivatives $f'$ and $f''=D^2f$ are taken on each edge with respect to the coordinates introduced above;
and let $q=\left\{ q_e \right\}_{e\in\E}$, with $q_e\in L^2(e)$ for each $e\in\E(\G)$ be a real-valued potential function.
Then the domain of $A$ and its action thereon are given by
\begin{eqnarray}
  \dom(A) &=& \left\{ f\in H^2(\G) : f \text{ satisfies \ref{Robin} } \forall v\in\V(\G) \right\}, \\
  (Af) (x) &=& -f''(x) + q(x) f(x)\,.
\end{eqnarray}
The Robin condition \ref{Robin} makes sense because $f\in H^2(\Gamma)$ has well defined derivatives at the endpoints of each edge.  That $A$ is self-adjoint in $L^2(\G)$ is subsumed by \cite[Theorem~1.4.4]{BerkolaikoKuchment2013}.

The periodicity of $A$ means that $\alpha_{gv} = \alpha_v$ and $q_{ge}(gx) = q_e(x)$ for all $v\in\V(\G)$ and $x\in e\in\E(\G)$ and for all $g\in\ZZ^n$.

The {\em Floquet modes} of $A$ (simultaneous eigenfunctions of $A$ and $\ZZ^n$) do not lie in $\dom(A)$, but in a larger space of functions that are locally like those in $\dom(A)$,
\begin{equation}
  \Hloc(\G) \,=\,
  \left\{ f=\{f_e\}_{e\in\E} : f \text{ is continuous; } f_e\in H^2(e) \,\forall e\in\E; \text{ $f$ satisfies \ref{Robin} } \forall v\in\V  \right\}.
\end{equation}

{\bfseries\slshape Reduction to a combinatorial graph.}
It is common to investigate the eigenvalue problem $(A-\lambda)u = 0$ for $u\in\Hloc(\G)$ by reducing it to an equivalent nonlinear-in-$\lambda$ eigenvalue problem $\Af(\lambda)\bar u = 0$ for a combinatorial graph, as long as $\lambda$ is not a Dirichlet eigenvalue for any edge, that is, $s_e(\lambda):=s_{q_e}(\lambda)\not=0$ for all $e\in\E(\G)$.
The {\em Dirichlet spectrum} of a quantum graph $(\G,A)$ is the set $\sigma_D(A)$ consisting of all the Dirichlet eigenvalues of all the edges.  For a periodic quantum graph for which a fundamental domain consists of a finite number of vertices and edges (that is, the set of orbits of the $\ZZ^n$ action is a finite graph), this set is discrete,
\begin{equation}\label{sigmaD}
  \sigma_D(A) \,=\, \left\{ \lambda\in\CC : \exists e\in\V(\G), s_e(\lambda) = 0 \right\}\,.
\end{equation}
For $\lambda\not\in\sigma_D(A)$, the equation $(A-\lambda)u = 0$ is equivalent to $\Af(\lambda)\bar u = 0$, where $\bar u$ is the restriction of $u$ to $\V(\G)$ and $\Af(\lambda)$ is a periodic ({\itshape i.e.}, $\ZZ^n$-invariant) operator that acts on functions defined on $\V(\G)$.  This reduction is accomplished by invoking the Dirichlet-to-Neumann map $G_{q_e}(\lambda)$ (\ref{Gq}) for each edge~$e$ to rewrite the Robin condition~\ref{Robin} solely in terms of the values of the function $u$ at~$v$ and all of its adjacent vertices.  One obtains
\begin{equation}\label{Robin2} 
  [\Af(\lambda)\bar u](v) \,:=\, \sum_{e=(v,w)\in\dE(v)} \frac{1}{s_e(\lambda)}\, \bar u(w)
  \,-\,  \left( \alpha_v + \sum_{e\in\dE(v)} \frac{c_e(\lambda)}{s_e(\lambda)} \right) \bar u(v)
  \,=\, 0\,.
\end{equation}
As a definition of the operator $\Af(\lambda)$, it is understood that $\bar u : \V(\G) \to \CC$ is arbitrary.
In conclusion, one has
\begin{equation}
  (A-\lambda)u = 0
  \quad\iff\quad
  \Af(\lambda)\bar u = 0\,.
\end{equation}
$\Af(\lambda)$ is called the reduced $\lambda$-dependent combinatorial operator associated with the quantum-graph operator~$A$.

\subsection{Floquet transform and Floquet surface}\label{sec:floquet}

The Floquet transform is the Fourier transform with respect to the $\ZZ^n$ action on the graph $\G$.  Given a function $f$, whose domain $\G$ includes points on the edges in the case of a metric graph (or just the vertex set of the graph in the case of a combinatorial graph), define its Floquet transform by
\begin{equation}
  \hat f(z,x) := \sum_{g\in\ZZ^n} f(gx)\, z^{-g}
  \qquad
  \text{for } z = (z_1,\dots,z_n) \in (\CC^*)^n,
\end{equation}
in which $z^h = z_1^{h_1}\cdots z_n^{h_n}$ for $h\in\ZZ^n$.  
This is a formal Laurent series in the symbol $z$ whose coefficients are shifts of $f$.  
The essential property of $\hat f$ is its quasi-periodicity in $x$,
\begin{equation}\label{qper}
  \hat f(z,gx) \,=\, \hat f(z,x)\, z^g\,,
\end{equation}
which makes $\hat f(z,\cdot)$ an eigenfunction for the $\ZZ^n$ action with eigenvalue $z^g$ for $g\in\ZZ^n$.

If $f\in L^2$, the Fourier inversion theorem holds,
\begin{equation}
  f(x) = \frac{1}{(2\pi)^n} \int_{\TT^n} \hat f(z,x)\, dV(z)\,,
\end{equation}
in which $dV$ is the $n$-dimensional volume measure on the $n$-torus
\begin{equation}
  \TT^n := \left\{ z\in\CC^* : |z_1|=\dots=|z_n|=1 \right\}.
\end{equation}
Property~\ref{qper} shows that $\hat f(z,\cdot)$ is determined by its restriction to a fundamental domain $W$ of $\G$.  When $\hat f$ is considered as a function of variables $z\in\TT^n$ and $x\in W$, the Floquet transform is a unitary transformation of Hilbert spaces,
\begin{equation}
  \hat{} \;:\;  L^2(\G) \to L^2(\TT^n;L^2(W))\,. 
\end{equation}
If $f$ is supported in a finite number of translations $gW$ of $W$, $\hat f(z,x)$ reduces to a Laurent polynomial.

If an operator $A$ commutes with the $\ZZ^n$ action, the Floquet transform converts $A$ into a multiplication operator in the variable $z$ in the sense that
\begin{equation}\label{FloquetA}
  (Af)\,\hat{} (z,\cdot) \,=\, \hat A(z) \hat f(z,\cdot),
\end{equation}
in which $\hat A(z)$ is a linear operator in $L^2(W)$.  If $A$ is a periodic quantum-graph operator, as described in the previous section, then the operator $\hat A(z)$ has the form $-D^2+q(x)$ as a differential operator, and it depends on $z$ only through its domain, which consists restrictions to $W$ of functions in the space $H^2_z(\G)$ of eigenfunctions for $\ZZ^n$,
\begin{equation}
  H^2_z(\G) := \left\{ f\in\Hloc(\G) : f(g\,\cdot) = z^g f \;\;\forall g\in\ZZ^n \right\}.
\end{equation}
Under the Floquet transform, the combinatorial graph operator $\Af(\lambda)$ becomes an operator $\hat\Af(\lambda,z)$.
This operator acts in the finite-dimensional complex vector space $L^2(\V(W))=\CC^{\V(W)}$, which consists of functions defined on the vertex set of a fundamental domain $W$ (so $\V(W)$ acts like a basis).  A typical element of $\CC^{\V(W)}$ is denoted by $\bar u=\left\{ u(v) \right\}_{v\in\V(W)}$.  Thus
\begin{equation}
  \hat\Af(\lambda,z) : \CC^{\V(W)} \to \CC^{\V(W)}.
\end{equation}
The operator $\hat\Af(\lambda,z)$ is a Laurent polynomial in $z$ with matrix-valued coefficients that are meromorphic in~$\lambda$ with poles at the Dirichlet eigenvalues of the edges.

For any fixed value of the spectral variable $\lambda\in\CC\setminus\sigma_D(A)$, the {\em Floquet surface} for $\lambda$ is defined to be the~set
\begin{equation}
  \Phi_\lambda \,:=\,
  \left\{ z\in(\CC^*)^n : \exists\,\text{nontrivial solution of } \hat\Af(\lambda,z) \bar u = 0 \right\}\,.
\end{equation}
If $\lambda\not\in\sigma_D(A)$, then the nontrivial solvability of $(\hat A(z)-\lambda)u = 0$ is equivalent to the nontrivial solvability of $\hat\Af(\lambda,z) \bar u = 0$.  This is in turn equivalent to the vanishing of the Laurent polynomial
\begin{equation}
  D(\lambda,z) := \det \hat\Af(\lambda,z),
\end{equation}
that is,
\begin{equation}
  \Phi_\lambda \,=\,
  \left\{ z\in(\CC^*)^n :  D(\lambda,z)=0 \right\}
  \qquad
  \text{for }\;
  \lambda\not\in\sigma_D(A)\,.
\end{equation}
Observe that the factorization of $D(\lambda,z)$ into eigenvalues of $\Af(\lambda,z)$ typically does not constitute reducibility because the eigenvalues generically are not polynomials in $z_1,\dots,z_n$.
Recall from the discussion after equation (\ref{Floquetsurface}) that the reducibility of $\Phi_\lambda$ is equivalent to the factorization of $D(\lambda,z)$ as a Laurent polynomial in the variables $z=(z_1,\dots,z_n)$ into two distinct nontrivial Laurent polynomial factors, $D(\lambda,z)=D_1(\lambda,z)D_2(\lambda,z)$.

\subsection{Constructing a bilayer periodic graph}

Let $(\Go,\Ao)$ be a periodic quantum graph, as described in the previous section, with vertex set $\Vo$, edge set~$\Eo$ and potential functions $\{q_e\}_{e\in\Eo}$.  Let a new periodic quantum graph $(\G,A)$ be constructed by connecting pairs of respective vertices of two disjoint copies of $(\Go,\Ao)$ with unit-length edges endowed with potentials that preserve the periodicity.  Precisely, the vertex set $\V=\V(\G)$ is the disjoint union
\begin{equation}
  \V := \Vo \sqcup \Vo
     = \Vo \times \left\{ 1,2 \right\},
\end{equation}
so an element of $\V$ is of the form $(v,1)$ or $(v,2)$, with $v\in\Vo$.
The edge set is
\begin{equation}
  \E := \Eo \sqcup \Eo \cup \E_c\,,
\end{equation}
in which the set of connecting edges is
\begin{equation}
  \E_c \,:=\, \left\{ e_v := \left\{ (v,1), (v,2) \right\} : v\in\Vo \right\}
\end{equation}
and one identifies an element $(\left\{ v,w \right\}, i) \in \Eo \sqcup \Eo$ with the element $\left\{ (v,i), (w,i) \right\}$ of pairs of vertices in $\V$.  The graph $\Go$ inherits the group action of $\ZZ^n$, so it is $n$-fold periodic.

Of course, the edges of $\G$ in the two copies of $\Go$ inherit the coordinates and potentials from $\Go$, that is, for $e=\left\{ (v,i),(w,i) \right\}$, one has $q_e=q_{\left\{ v,w \right\}}$, which naturally allows one to use the coordinate $x_{(v,w)}$ for the directed edge $\left( (v,i),(w,i) \right)$.  Each connecting edge $e\in\E_c$ is endowed with a coordinate in $[0,1]$ (so $L_e=1$) and a potential function $q_e$.  These potentials are assumed to be periodic, that is $q_{ge}(gx) = q_e(x)$ for all $g\in\ZZ^n$ and all $e\in\E_c$.  For the connecting edge $\left\{ (v,1), (v,2) \right\}$ associated with the vertex $v\in\Go$, it is convenient to denote the potential by $q_v$,
\begin{equation}
  q_{\left\{ (v,1), (v,2) \right\}} \,=\, q_v\,.
\end{equation}
The spectral functions $s_e(\lambda)$ for undirected connecting edges and the functions $c_e(\lambda)$ for directed connecting edges are denoted by
\begin{equation}
  \renewcommand{\arraystretch}{1.1}
\left.
\begin{array}{l}
  s_{\left\{ (v,i),(v,\itilde) \right\}} = s_v(\lambda) \\
  c_{\left( (v,i),(v,\itilde) \right)} = c_{v,i}(\lambda)
\end{array}
\right\}
  \qquad
  i,\itilde \in\left\{ 1,2 \right\},\; i\not=\itilde.
\end{equation}
The $a$ and $b$ functions for the directed connecting edge $((v,1),(v,2))$ are
\begin{equation}
  \renewcommand{\arraystretch}{1.1}
\left.
\begin{array}{l}
  a_v(\lambda) = \half \left( c_{v,1}(\lambda) - c_{v,2}(\lambda) \right) \\
  b_v(\lambda) = \half \left( c_{v,1}(\lambda) + c_{v,2}(\lambda) \right)
\end{array}
\right..
\end{equation}
The Dirichlet-to-Neumann matrix $G_v(\lambda)$ for $((v,1),(v,2))$ is given by
\begin{equation}\label{Gv}
  s_v(\lambda) G_v(\lambda) =
\renewcommand{\arraystretch}{1.2}
\left[
  \begin{array}{cc}
    -c_{v,1}(\lambda) & 1 \\
    1 & -c_{v,2}(\lambda) 
  \end{array}
\right]
\,=\,
- \left[
  \begin{array}{cc}
    b_v(\lambda) & 0 \\
    0 & b_v(\lambda) 
  \end{array}
\right]
+
\left[
  \begin{array}{cc}
    -a_v(\lambda) & 1 \\
    1 & a_v(\lambda) 
  \end{array}
\right]\,.
\end{equation}

Having the metric graph $\G$ together with the potentials on its edges, section \ref{sec:periodicqg} defines a self-adjoint operator $A$ in $L^2(\G)$, with the same Robin condition \ref{Robin}.  This operator commutes with the $\ZZ^n$ action because of the periodicity of the potentials $q_e$.

\begin{definition}[bilayer graph]\label{def:bilayer}
  Let be given a quantum graph $(\Go,\Ao)$ with a $\ZZ^n$ symmetry group and potentials $\{ q_v : v\in\V(\Go) \}$ defined on edges $\left\{ (v,1), (v,2) \right\}$ of the the disjoint union $\V(\Go) \sqcup \V(\Go)$.  Let the potentials $q_v$ be invariant under $\ZZ^n$.  The periodic quantum graph $(\G,A)$ obtained by connecting two copies of $(\Go,\Ao)$ by edges of unit length with the potentials $q_v$, as described in this section, is called the {\em bilayer quantum graph} associated to $(\Go,\Ao)$ and the potentials $\{ q_v : v\in\V(\Go) \}$.
\end{definition}

{\bfseries\slshape Reduction to a combinatorial graph.}
In the formulation of the eigenvalue problem $(A-\lambda)u=0$ for a bilayer quantum graph in terms of the restriction $\bar u$ of $u$ to the vertex set $\V$, the expression in (\ref{Robin2}) has to be augmented, for the $i^\text{th}$ copy of $\Vo$ ($i\in\left\{ 1,2 \right\}$), by adding to it the following terms coming from the connecting edges:
\begin{equation}
  - \frac{c_{v,i}(\lambda)}{s_v(\lambda)} \bar u((v,i)) + \frac{1}{s_v(\lambda)} \bar u((v,\itilde)) 
  \qquad
  i,\itilde \in\left\{ 1,2 \right\},\; i\not=\itilde.
\end{equation}
This is valid whenever $\lambda\not\in\sigma_D(A)$.
Define operators that are diagonal with respect to $\Vo$, for $i=1,2$,
\begin{equation}
  \Cf(\lambda) := \mathrm{diag} \left( \frac{1}{s_v(\lambda)} \right)\,,
  \qquad
  \Cf_i(\lambda) := \mathrm{diag} \left( \frac{c_{v,i}(\lambda)}{s_v(\lambda)} \right)\,,
\end{equation}
that is, $\left[ \Cf(\lambda) \bar u \right](v) = s_v(\lambda)^{-1} \bar u(v)$, and similarly for $\Cf_i(\lambda)$.
Then the operator $\Af(\lambda)$ has the following block form with respect to the decomposition $\V=\Vo\sqcup\Vo$:
\begin{equation}\label{Alambda}
  \Af(\lambda) 
\equiv
\renewcommand{\arraystretch}{2}
\left[
\begin{array}{cc}
  \Afo(\lambda) & 0 \\
  0  &    \Afo(\lambda)
\end{array}
\right]
+
\renewcommand{\arraystretch}{2}
\left[
\begin{array}{cc}
  -\Cf_1(\lambda) & \Cf(\lambda) \\
  \Cf(\lambda)  &  -\Cf_2(\lambda)
\end{array}
\right]
.
\end{equation}
Under the Floquet transform, this operator becomes
\begin{equation}\label{Alambdaz}
  \hat\Af(\lambda,z)
\equiv
\renewcommand{\arraystretch}{2}
\left[
\begin{array}{cc}
  \hat\Afo(\lambda,z) & 0 \\
  0  &  \hat\Afo(\lambda,z)
\end{array}
\right]
+
\renewcommand{\arraystretch}{2}
\left[
\begin{array}{cc}
  -\Cf_1(\lambda) & \Cf(\lambda) \\
  \Cf(\lambda)  &  -\Cf_2(\lambda)
\end{array}
\right].
\end{equation}
The operator
\begin{equation*}
\Gf(\lambda)
=
\renewcommand{\arraystretch}{1.7}
\left[
\begin{array}{cc}
  -\Cf_1(\lambda) & \Cf(\lambda) \\
  \Cf(\lambda)  &  -\Cf_2(\lambda)
\end{array}
\right]  
\end{equation*}
is the ``energy-dependent coupling operator".
It does not depend on $z$ because the connecting edges do not connect vertices in two different $\ZZ^n$-translates of a fundamental domain $W$.  When considered as being applied to functions $u|_{\V(W)}$ restricted to the vertices in $W\!$, they become finite diagonal matrices with respect to the natural basis of complex functions on the finite set $\V(W) = \V(\Wo)\sqcup\V(\Wo)$.

\section{Coupling by symmetric edges}\label{sec:symmetric}

When the potentials $q_v$ on the connecting edges of the bilayer graph $\G$ are symmetric about their midpoints, $A$ has reflectional symmetry.  This means that $A$ commutes with the reflection in $\G$ which maps each copy of $\Go$ to the other and reflects each connecting edge about its center.  Thus $A$ is invariant on the eigenspaces $\Hh_+$ and $\Hh_-$ of this reflection; the former consists of functions that are even with respect to the reflection (eigenvalue $1$), and the latter consists of odd functions (eigenvalue $-1$).  With respect to the orthogonal decomposition $L^2(\G)=\Hh_+\oplus\Hh_-$, the operator $A$ has a decomposition
\begin{equation}
  A = A_+ + A_-\,.
\end{equation}
This kind of symmetric periodic quantum graph is studied in~\cite[\S3.2]{Shipman2014}.  Both $A_+$ and $A_-$ are unitarily equivalent to quantum-graph operators $\tilde A_+$ and $\tilde A_-$ defined on the metric graph $\G_d$ that is ``half'' of $\G$, consisting of $\Go$ plus a dangling edge attached to each vertex.  For $v\in\V(\Go)$, this dangling edge is half of the connecting edge~$e_v$.
By restricting functions in $\dom(A_+)=\dom(A)\cap\Hh_+$ to $\G_d$, one obtains the domain of $\tilde A_+$, which possesses the Neumann boundary condition $du/dx=0$ at the free vertices of the dangling edges.  And the restriction of functions in $\dom(A_-)=\dom(A)\cap\Hh_-$ to $\G_d$ is the domain of $\tilde A_-$; it possesses the Dirichlet condition $u=0$ at the free vertices.

This decomposition of $A$ renders its Floquet surface canonically reducible.  Let $\Af_+(\lambda)$ and $\Af_-(\lambda)$ be the reduced $\lambda$-dependent combinatorial operators for $A_+$ and $A_-$ and $\Af_+(\lambda,z)$ and $\Af_-(\lambda,z)$ their Floquet transforms.  The Floquet surface for $A$ at energy $\lambda$ is just
\begin{equation}
  \Phi_\lambda = \left\{ z\in(\CC^*)^n :\; \right( \det\Af_+(\lambda,z) \left) \right( \det\Af_-(\lambda,z) \left) \,= 0 \right\}\,,
\end{equation}
which is reducible into the union of the Floquet surfaces of the two quantum-graph operators $\tilde A_+$ and $\tilde A_-$.

\section{Coupling within an asymmetry class}\label{sec:asymmetric1}

If the potentials connecting two copies of the periodic quantum graph $\Go$ are not symmetric, it is no longer possible to decompose the bilayer-graph Schr\"odinger operator $A$ into two components unitarily equivalent to decorated quantum graphs, as described in the previous section.  But if the potentials belong to same asymmetry class, it turns out that the reduced $\lambda$-dependent combinatorial operator $\Af(\lambda)$ can be decomposed.  Its invariant subspaces depend on $\lambda$, so this reduction does not proceed from a reduction of $A$.  Nevertheless, the energy-dependent reduction does lead to a factorization of the determinant $D(\lambda,z)$ as a Laurent polynomial in $z=(z_1,\dots,z_n)$ and therefore to the reducibility of the Floquet surface for all energies~$\lambda$.

\smallskip

The simple mechanism for generating two components of the Floquet surface when the connecting edges are in the same asymmetry class is discussed in the following two paragraphs.  The arguments are made rigorous in this section and the following one.  One tries to construct a Floquet mode $f$ for a bilayer graph at a given $\lambda$ by first defining it on any edge connecting two copies of a given vertex $v$ of the single layer (call them $v_1$ and $v_2$ in the two layers) to be a special solution of $(-d^2/dx^2 + q_v(x) - \lambda)u_v(x)=0$.  The solution $u_v(x)$ is required to be such that the vector $[u_v(v_1),u_v(v_2)]^t$ of endpoint values is an eigenvector of the DtN matrix $G_v(\lambda)$ with eigenvalue $\mu_v$ (it is defined up to a constant multiple).  This ensures that the same relation between value and derivative occurs at each endpoint $v_1$ and $v_2$, that is, $u_v'(v_1)=\mu_v u(v_1)$ and $u_v'(v_2)=\mu_v u(v_2)$.  Now the Robin vertex condition~\ref{Robin} for the bilayer graph can be reduced to a Robin condition for each single layer, with $\alpha$ replaced by $\alpha-\mu_v$ applied to both copies $v_1$ and $v_2$ of the vertex $v$.
This modified single-layer graph has its own Floquet surface $\tilde D(\lambda,z)=0$.
The Floquet mode $f$ we are seeking, restricted to each of the two copies of the single layer, will be equal to a constant multiple of some Floquet mode for the modified single layer, for a value of $z$ satisfying $\tilde D(\lambda,z)=0$, and the constant multiple will be different in the two layers, say $c_1$ and $c_2$.
In order for the modes in the two layers to be compatible with the functions $u_v(x)$ on the connecting edges, the vector of constants $[c_1,c_2]^t$ must be proportional to each of the vectors of endpoint values $[u_v(v_1),u_v(v_2)]^t$.  This means that $[c_1,c_2]^t$ must be a common eigenvector for all the matrices $G_v(\lambda)$, over all vertices $v$ of~$\Go$.  So the Floquet surface $\tilde D(\lambda,z)=0$ for the modified single layer is one component of the Floquet surface of the bilayer graph.  The other component is obtained by using the other eigenvector of $G_v(\lambda)$.

This construction of the components of the Floquet surface is therefore possible whenever the matrices $G_v(\lambda)$, for all vertices $v$ in $\Go$, are simultaneously diagonalizable, or whenever they commute with each other.  This happens for all $\lambda$ exactly when all the A-functions $a_v(\lambda)$ are identical.  To see this, write the DtN map~as
\begin{equation*}
  s_v(\lambda) G_v(\lambda) \,=\,
  -\renewcommand{\arraystretch}{1.2}
\left[
\begin{array}{cc}
  \!\!b_v(\lambda) & 0 \\
  0 & b_v(\lambda)\!\!
\end{array}
\right]
  +\renewcommand{\arraystretch}{1.2}
\left[
\begin{array}{cc}
  \!\!-a_v(\lambda) & 1 \\
  1 & a_v(\lambda)\!\!
\end{array}
\right],
\end{equation*}
and denote
\begin{equation}
  N_v(\lambda) = \left[
\begin{array}{cc}
  \!\!-a_v(\lambda) & 1 \\
  1 & a_v(\lambda)\!\!
\end{array}
\right].
\end{equation}
Two matrices $G_v(\lambda)$ and $G_w(\lambda)$ commute whenever $N_v(\lambda)$ and $N_w(\lambda)$ commute, and a calculation shows that this happens if and only if $a_v(\lambda)=a_w(\lambda)$, as stated in part (4) of Theorem~\ref{thm:a}.

\subsection{Main theorem: Reducibility of the Fermi surface}

Let the potentials of all of the connecting edges of a bilayer graph $\G$ possess the same A-function,
\begin{equation}
  a_v(\lambda) \,=\, a(\lambda)
  \qquad
  \forall\, v \in \V(\Wo),
\end{equation}
so that
\begin{equation}\label{sG}
  s_v(\lambda) G_v(\lambda) \,=\,
  -\renewcommand{\arraystretch}{1.2}
\left[
\begin{array}{cc}
  \!\!b_v(\lambda) & 0 \\
  0 & b_v(\lambda)\!\!
\end{array}
\right]
  +\renewcommand{\arraystretch}{1.2}
\left[
\begin{array}{cc}
  \!\!-a(\lambda) & 1 \\
  1 & a(\lambda)\!\!
\end{array}
\right],
\end{equation}
which can be decomposed by means of the spectral resolution of the second matrix
\begin{equation}
  s_v(\lambda) G_v(\lambda) \,=\, -b_v(\lambda) I_2 + \mu P_{\!\mu} - \mu P_{\!-\mu}
\end{equation}
with $\mu^2 = a(\lambda)^2+1$.

When the potentials on the connecting edges are symmetric and therefore $a(\lambda)=0$, one has $\mu^2=1$ and the eigenvalues become constant.
The Riemann surface $\tilde S$ separates into two copies of the $\lambda$-plane, one with $\mu=1$ and one with $\mu=-1$.
Since the eigen-projections $P_1$ and $P_{\!-1}$ are now constant, they provide a decomposition of the full quantum-graph operator $A$.  This is evidenced in the $\mu$-dependence of the diagonal matrices $\Df^\pm(\lambda)$ in the following theorem.

\begin{theorem}\label{thm:asymmetric1}  
  Let $(\Go,\Ao)$ be a periodic quantum graph and $(\G,A)$ a bilayer quantum graph obtained by joining two copies of $(\Go,\Ao)$ according to Definition~\ref{def:bilayer}. Let the potentials of the connecting edges lie in the asymmetry class associated with a single A-function $a(\lambda)$.
  
  The Floquet surface $\Phi_\lambda$ of $(\G,A)$ is reducible for all $\lambda\notin\sigma_D(A)$ into the union of two surfaces,
\begin{equation}
  \Phi_\lambda =  \Phi^+_\lambda \cup \Phi^-_\lambda,
\end{equation}
in which
\begin{equation}
  \Phi^\pm_\lambda = \left\{ z\in(\CC^*)^n : \det\big( \hat\Afo(\lambda,z) + \Df^\pm(\lambda) \big) = 0 \right\},
\end{equation}
the diagonal matrices $\Df^\pm(\lambda)$ are defined by
\begin{equation}
  \Df^\pm(\lambda) := \mathrm{diag}_{v\in\Wo} \left( \frac{-b_v(\lambda)\pm\mu}{s_v(\lambda)} \right)\,,
\end{equation}
with $\mu^2 = a(\lambda)^2+1$, and\, $\hat\Afo(\lambda,z)$ is the Floquet transform of the combinatorial reduction of $\Ao-\lambda I$.
Both components $\Phi^+_\lambda$ and $\Phi^-_\lambda$ are nonempty subsets of $(\CC^*)^n$ for all but possibly a discrete set of values of $\lambda$.
\end{theorem}

\begin{proof}  
The Floquet surface of $(\Gamma,A)$ is the zero set of the determinant of the matrix in~(\ref{Alambdaz}).  That matrix can be expressed as
\begin{align}
  \Af(\lambda,z) &\;=\;
  I_2\otimes\hat\Afo(\lambda,z) - I_2\otimes\Cf(\lambda)\Bf(\lambda) + N(\lambda)\otimes \Cf(\lambda) \\
        &\;=\;
  I_2\otimes\Cf(\lambda)
  \left[
  I_2\otimes \left( \Cf(\lambda)^{-1}\hat\Afo(\lambda,z) - \Bf(\lambda) \right)
     + N(\lambda)\otimes I_m  
  \right],
\end{align}
in which $m$ is the number of vertices in a fundamental domain of $\Go$ and
\begin{equation}
  \Bf(\lambda) := \mathrm{diag}\left( b_v(\lambda) \right),
  \qquad
  N(\lambda) =
  \renewcommand{\arraystretch}{1.1}
   \left[
   \begin{array}{cc}
        -a(\lambda) & 1 \\
        1 & a(\lambda)
   \end{array}
   \right],
\end{equation}
with $a(\lambda)$ being the common A-function for the potentials $q_v$ of all the connecting edges.  The eigenvalues of $N(\lambda)$ are $\pm\mu$, and thus by conjugating $N(\lambda)$, the determinant of $\Af(\lambda,z)$ can be expressed as
\begin{align}
   \det \Af(\lambda,z)
   &=
   \det\Cf(\lambda)^2
      \det\left[ I_2\otimes\left( \Cf(\lambda)^{-1}\Afo(\lambda,z)-\Bf(\lambda) \right) +
        \renewcommand{\arraystretch}{0.9} \left[\hspace{-5pt} \begin{array}{cc} \mu\!\!&0\\0\!\!&-\mu \end{array} \hspace{-3pt}\right]\otimes I_m \right] \\ \label{dets}
   &= 
   \det\Cf(\lambda)^2
    \det\left( \Cf(\lambda)^{-1}\Afo(\lambda,z)-\Bf(\lambda)+\mu I_m \right)
    \det\left( \Cf(\lambda)^{-1}\Afo(\lambda,z)-\Bf(\lambda)-\mu I_m \right) \\
   &= 
    \det\left( \Afo(\lambda,z) + \Cf(\lambda)\left(-\Bf(\lambda)+\mu I_m\right) \right)
    \det\left( \Afo(\lambda,z) + \Cf(\lambda)\left(-\Bf(\lambda)-\mu I_m\right) \right).
\end{align}
The matrices $\Df^\pm(\lambda)$ in the statement of the theorem are equal to
$\Cf(\lambda)\left( -\Bf(\lambda)\pm\mu I_m \right)$, so that
\begin{equation}\label{det=detdet}
  \det \hat\Af(\lambda,z) =
  \det \left( \hat\Afo(\lambda,z) + \Df^+(\lambda) \right)
  \det \left( \hat\Afo(\lambda,z) + \Df^-(\lambda) \right)\,.
\end{equation}
Each of these factors is a Laurent polynomial in $z=(z_1,\cdots,z_n)$ because the matrix $\hat\Afo(\lambda,z)$ has entries that are Laurent polynomials.
Since the Floquet surface for $A$ at energy $\lambda$ is
\begin{equation}
  \Phi_\lambda \,=\,
  \left\{ z\in (\CC^*)^n : \det \hat\Af(\lambda,z) = 0 \right\},
\end{equation}
the Floquet surface is reducible into the two components stated in the theorem.
Denote the determinant factors in~(\ref{det=detdet})by $D^+(\lambda,z)$ and $D^-(\lambda,z)$.

Let $a^+_{r_1,\dots,r_n}(\lambda)$ and $a^-_{r_1,\dots,r_n}(\lambda)$ be the coefficients of $D^+(\lambda,z)$ and $D^-(\lambda,z)$ for a given monomial $z_1^{r_1}\cdots z_n^{r_n}$.   One of these must be nonzero for some choice of $r_1,\dots,r_n$, not all zero.
Observe that $D^+(\lambda,z)$ and $D^-(\lambda,z)$ are distinguished only by the sign of $\mu=\sqrt{a(\lambda)^2+1\,}$, so that also $a^+_{r_1,\dots,r_n}(\lambda)$ and $a^-_{r_1,\dots,r_n}(\lambda)$ are distinguished only by the sign of $\mu$.
This means that $a^+_{r_1,\dots,r_n}(\lambda)$ and $a^-_{r_1,\dots,r_n}(\lambda)$ are actually equal to the two branches over the $\lambda$-plane of a single analytic function $a_{r_1,\dots,r_n}(\lambda)$ on the Riemann surface $\Ss$.  This function vanishes only at a discrete set of points.  Therefore, the two components of the Floquet surface are nonempty.
\end{proof}

That the two components of $\Phi_\lambda$ are distinct for all but a discrete set of values of $\lambda$ follows from the determinants of the matrices $\Cf(\lambda)^{-1}\Afo(\lambda,z)-\Bf(\lambda)\pm\mu I_m$ depending on the sign in front of $\mu$.  This is apparently a generic condition.

\subsection{Energy-dependent decomposition of a bilayer graph}

The discussion at the beginning of section~\ref{sec:asymmetric1} points to an energy-dependent decomposition of the combinatorial reduction of a periodic bilayer quantum graph as the mechanism for reducibility of its Floquet surface.
This section elucidates that mechanism and results in an alternative proof of Theorem~\ref{thm:asymmetric1}.

As described in section~\ref{sec:floquet}, $\hat\Af(\lambda,z)$ acts in the space $\CC^{\V(W)}$.  The block form~(\ref{Alambdaz}) is with respect to the decomposition
\begin{equation}
  \CC^{\V(W)} = \CC^{\V(\Wo)} \times \CC^{\V(\Wo)} \cong \CC^2 \otimes \CC^{\V(\Wo)}.
\end{equation}
This identification with the tensor product $\CC^2 \otimes \CC^{\V(\Wo)}$ provides a convenient way to express the operator~$\hat\Afo(\lambda,z)$.
The block operator $\Gf(\lambda)$ defined in~(\ref{Alambdaz}) has diagonal blocks, so it is a sum of tensor products.  Let $E_v$ denote the projection in $\CC^{\V(\Wo)}$ to the $v$-component (that is, $[E_v(u)](w) = u(v)\delta_{vw}$, where $\delta_{vw}$ is the Kronecker symbol); then
\begin{equation}\label{calG}
  \Gf(\lambda) = \sum_{v\in\V(\Wo)} G_v(\lambda) \otimes E_v\,.
\end{equation}
The operators $\Af(\lambda)$ and $\hat\Af(\lambda,z)$ can now be written
\begin{eqnarray}
  \Af(\lambda) &=& I_2 \otimes \Afo(\lambda) + \sum_{v\in\V(\Wo)} G_v(\lambda) \otimes E_v\,, \label{Aflambda} \\
  \hat\Af(\lambda,z) &=& I_2 \otimes \hat\Afo(\lambda,z) + \sum_{v\in\V(\Wo)} G_v(\lambda) \otimes E_v\,,
\end{eqnarray}
in which $I_2$ is the identity $2\times2$ matrix.

The projections $P_\mu$ are independent of the vertex $v$ since $a_v(\lambda)=a(\lambda)$ for all $v\in\Wo$.  Thus the coupling operator 
\begin{equation}\label{Gf}
  \Gf(\lambda) \,=\, \sum_{v\in\V(\Wo)} G_v(\lambda) \otimes E_v
\end{equation}
in (\ref{Alambdaz},\ref{calG}) is resolved with respect to the projections $P_{\!\pm\mu}$ as
\begin{equation}\label{Gf2}
\begin{split}
  \Gf(\lambda) &=\, \sum_{v\in\V(\Wo)}  \frac{1}{s_v(\lambda)}
  \left( -b_v(\lambda) I_2 + \mu \big(P_{\!\mu}(\lambda)-P_{\!-\mu}(\lambda)\big) \right)
       \otimes E_v  \\
       &=\; P_{\!\mu} \otimes\! \sum_{v\in\V(\Wo)} \frac{-b_v(\lambda)+\mu}{s_v(\lambda)} E_v
          \;+\; P_{\!-\mu} \otimes\! \sum_{v\in\V(\Wo)} \frac{-b_v(\lambda)-\mu}{s_v(\lambda)} E_v\,.
\end{split}
\end{equation}
The resolution of $\Af(\lambda)$ with respect to the projections $P_{\!\pm\mu}$ is obtained from this resolution of $\Gf(\lambda)$ and equation (\ref{Aflambda}) (using the relation $I_2 = P_{\!\mu}+P_{\!-\mu}$):
\begin{equation}\label{Gf3}
  \Af(\lambda) \;=\; 
       P_{\!\mu} \otimes\!
       \left( \Afo(\lambda) + \sum_{v\in\V(\Wo)} \frac{-b_v(\lambda)+\mu}{s_v(\lambda)} E_v \right)
          \;+\; P_{\!-\mu} \otimes\! 
       \left( \Afo(\lambda) + \sum_{v\in\V(\Wo)} \frac{-b_v(\lambda)-\mu}{s_v(\lambda)} E_v \right).
\end{equation}
To write $\Af(\lambda)$ in block form with respect to the projections $P_{\!\pm\mu}$ and $E_v$, set
\begin{equation}\label{Df}
  \Df^+(\lambda) := \mathrm{diag}_{v\in\Wo} \left( \frac{-b_v(\lambda)+\mu}{s_v(\lambda)} \right),
  \qquad
  \Df^-(\lambda) := \mathrm{diag}_{v\in\Wo} \left( \frac{-b_v(\lambda)-\mu}{s_v(\lambda)} \right),
\end{equation}
and the block form is
\begin{equation}\label{Alambdablock}
  \Af(\lambda)
\equiv
\renewcommand{\arraystretch}{2}
\left[
\begin{array}{cc}
  \Afo(\lambda) + \Df^+(\lambda) & 0 \\
  0  &  \Afo(\lambda) + \Df^-(\lambda)
\end{array}
\right].
\end{equation}
Since the connecting edges do not connect vertices in two distinct $\ZZ^n$-translates of a fundamental domain $W$ of $\G$, $\hat\Df(\lambda,z)= \Df(\lambda)$ as an operator in $\CC^{\V(W)}$, and thus
\begin{equation}\label{Alambdablockz}
  \hat\Af(\lambda,z)
\equiv
\renewcommand{\arraystretch}{2}
\left[
\begin{array}{cc}
  \hat\Afo(\lambda,z) + \Df^+(\lambda) & 0 \\
  0  &  \hat\Afo(\lambda,z) + \Df^-(\lambda)
\end{array}
\right].
\end{equation}
The factorization~(\ref{det=detdet}) of $\det \hat\Af(\lambda,z)$ follows from this block-diagonal form.

\subsection{Realizability of spectral components as graphs}\label{sec:realizability}

It is interesting to ask whether the components of the Floquet surface of a periodic bilayer quantum graph $(\G,A)$ can be realized as the Floquet surfaces of other, simpler, periodic bilayer quantum graphs.
More specifically, one can ask if the two operators $\Afo(\lambda) + \Df^\pm(\lambda)$ in~(\ref{Alambdablock}) can be realized as the $\lambda$-dependent combinatorial reduction on $\V(\Go)$ of a quantum graph obtained as a ``decoration" of $\Go$, that is, by attaching a dangling edge to each vertex $v$, where the edge dangling from vertex $v$ has potential $q^d_v(x)$, and these potentials commute with the $\ZZ^n$ action.  Section~\ref{sec:symmetric} describes this realization when the potentials on the connecting edges are symmetric about the midpoints.
Theorem~\ref{thm:realizability} says that such a realization of the components is not possible when the potentials are not symmetric, provided that the A-function satisfies a certain genericity condition.

\begin{condition}[genericity of the A-function]\label{def:genericity}
  The A-function $a(\lambda)$ for a potential $q\in L^2[0,1]$ satisfies the {\em genericity condition} if there exists $\lambda_0\in\CC$ such that $a(\lambda_0)^2+1=0$ and
\begin{equation}
  \int_0^1 \psi(x)^2 dx \;\not=\; 0,
\end{equation}
in which $-\psi''+(q(x)-\lambda_0)\psi=0$ and $[\psi(0),\psi(1)]$ is an eigenvector for the DtN matrix $G(\lambda_0)$ for $q$.
\end{condition}

This condition requires some explanation.  The square roots $\pm\mu$ of $a(\lambda)^2+1$ are involved in the combinatorial graph operators $\Afo(\lambda) + \Df^\pm(\lambda)$~(\ref{Alambdablock}) associated to the two components of the Floquet surface of $(\G,A)$.  The Riemann surface for $\sqrt{a(\lambda)^2+1\,}$ (section~\ref{sec:riemannsurface}) is connected if the derivative of $a(\lambda)^2+1$ at some root of $a(\lambda)^2+1$ does not vanish.  In fact, $a(\lambda)^2+1$ has infinitely many roots: If $a(\lambda)$ does not identically vanish, then neither does the entire function $a(\lambda)^2+1$.  Furthermore, $2a(\lambda)=c(\lambda)-\tilde c(\lambda)$ is of growth order $1/2$, and therefore so is $a(\lambda)^2+1$.  By the Hadamard theorem, $a(\lambda)^2+1$ is determined up to a constant multiple by its roots,
\begin{equation}
  a(\lambda)^2+1 \;=\; c \prod_{k=1}^\infty \big(1-\lambda_k^{-1} \lambda\big),
\end{equation}
and therefore must have infinitely many roots since it is not a polynomial.  Because of part (6) of Theorem~\ref{thm:a}, the genericity condition implies that $d(a(\lambda)^2+1)/d\lambda$ is nonvanishing at one of the roots of $a(\lambda)^2+1$.

\begin{theorem}[realizability of components]\label{thm:realizability}
  Let $(\G,A)$ be a periodic bilayer quantum graph associated to the periodic quantum graph $(\Go,\Ao)$ and coupling potentials $\{ q_v : v\in\V(\Go) \}$ that lie in the same asymmetry class, corresponding to A-function $a(\lambda)$.
\begin{enumerate}
\item
If the potentials $q_v$ are symmetric about the midpoints of the connecting edges, then the components
$\Afo(\lambda) + \Df^\pm(\lambda)$ of the combinatorial reduction of $(\G,A)$ can be realized as the combinatorial reductions of two different decorations of $(\Go,\Ao)$, one obtained by attaching a dangling edge to each vertex with the Dirichlet condition ($u(w)=0$) at each terminal vertex $w$, and the other obtained similarly with the Neumann condition ($u'(w)=0$) at each terminal vertex.
\item 
If the potentials $q_v$ are not symmetric about the midpoints of the connecting edges and $a(\lambda)$ satisfies the genericity Condition~\ref{def:genericity}, then neither component $\Afo(\lambda) + \Df^\pm(\lambda)$ of the combinatorial reduction of $(\G,A)$ can be realized as the combinatorial reduction of any decoration of $(\Go,\Ao)$ obtained by attaching a dangling graph to each vertex with self-adjoint vertex conditions.
\end{enumerate}
\end{theorem}

\begin{proof}
Part (1) is a result of section~\ref{sec:symmetric}.
Suppose that the potentials $q_v$ are not symmetric.
The operator $\Afo(\lambda)+\Df^+(\lambda)$ acts on functions $f:\V(\Gamma)\to\CC$ whose restriction to the vertices of the edge $e_v=\left\{ (v,1), (v,2) \right\}$ lies in the one-dimensional image of the eigen-projection $P_\mu$.  In other words, the domain of $\Afo(\lambda)+\Df^+(\lambda)$ consists of functions of the form $\psi_\mu \otimes g$, where $\psi_\mu\in\CC^2$ spans the image of $P_\mu$ and $g:\V(\Go)\to\CC$ is in the domain of $\Afo(\lambda)$.  This means that $f$ is completely determined by its restriction to one of the copies of~$\V(\Go)$.  Its restriction to one copy is just a scalar multiple of its restriction to the other, that multiple being equal to the ratio of the components of the eigenvector $\psi_\mu$ of $N(\lambda)$.

The self-adjoint vertex conditions of the graph dangling from vertex $v\in\V(\Go)$ imposes a $\lambda$-dependent Dirichlet-to-Neumann condition at~$v$,
\begin{equation}
  u'(v) \,=\, -m_v(\lambda)\, u(v)
\end{equation}
(the function $m_v(\lambda)$ is a Weyl-Titchmarsh M-function).  This has the effect of adding the term $m_v(\lambda)\bar u(v)$ to the expression defining $\Af(\lambda)$ (which is to be replaced by $\Afo(\lambda)$ in the present context) in~\ref{Robin2}.  The combinatorial reduction for the decorated graph, defined on $\V(\Go)$ is therefore
\begin{equation}
  \Afo(\lambda) + \Mf(\lambda),
\end{equation}
in which
\begin{equation}
  \Mf(\lambda) \,=\, \mathrm{diag}\left( m_v(\lambda) \right).
\end{equation}
The functions $m_v(\lambda)$ are meromorphic functions of $\lambda$ having all of their poles on the real axis.

By comparing $\Mf(\lambda)$ to $\Df^+(\lambda)$ (or $\Df^-(\lambda)$) as given in~(\ref{Df}), one sees that the requirement that
\begin{equation}
  \Afo(\lambda) + \Df^+(\lambda) \,=\, \Afo(\lambda) + \Mf(\lambda)
\end{equation}
implies that the functions $(-b_v(\lambda)+\mu)/s_v(\lambda)$ in~(\ref{Df}) must be meromorphic in $\CC$.
But this is possible only if the Riemann surface $\Ss$ for $\mu=\sqrt{a(\lambda)^2+1\,}$ has no ramification points over $\lambda$.

Let $\lambda_0$ be as provided by Condition~\ref{def:genericity}, and consider
$d(a(\lambda)^2+1)/d\lambda=2a(\lambda)da(\lambda)/d\lambda$ evaluated at $\lambda=\lambda_0$.  Since $a(\lambda_0)^2+1=0$, one has $a(\lambda_0)=i$, and by part (2) of Theorem~\ref{thm:a}, $\lambda_0\not\in\RR$ so that $s(\lambda_0)\not=0$.  This, together with Condition~\ref{def:genericity} and part (6) of Theorem~\ref{thm:a} yield that $da/d\lambda$ does not vanish at~$\lambda_0$.  Therefore $\partial(\mu^2-(a(\lambda)^2+1))/\partial\lambda$ does not vanish at $(\lambda=\lambda_0,\mu=0)$, whereas $\partial(\mu^2-(a(\lambda)^2+1))/\partial\mu=0$ at $(\lambda=\lambda_0,\mu=0)$, which means that $\Ss$ has a ramification point over $\lambda$ at $\lambda_0$.
This implies that $\Df^+(\lambda)$ is not analytic in $\lambda$ in $\CC\setminus\RR$, whereas $\Mf(\lambda)$ is, and therefore the component $\Afo(\lambda) + \Df^+(\lambda)$ cannot be realized as a decoration of $(\Go,\Ao)$.
This argument applies verbatim to $\Afo(\lambda) + \Df^-(\lambda)$ also.
\end{proof}

If the components of the combinatorial reduction of $(\G,A)$ have ramification points over $\lambda$, then the determinants in (\ref{dets}) will have the same ramification points, unless they just happen to be functions of~$\mu^2$ and $\lambda$.  So, barring this happening, the components of the Floquet surface cannot be realized as Floquet surfaces of decorated graphs.
It appears that, for any fixed value of the energy $\lambda$, the Floquet surface offers little insight into realizability of its components.  The obstruction to realizability is the branch points in the $\lambda$ variable.  This question is related to that of the determination of the graph operator from its dispersion function, which is understood for discrete graph Laplacians~\cite{GiesekerKnorrerTrubowitz1993}.  There is substantial literature on how the properties of metric graphs and Schr\"odinger operators on them are determined by various spectral data~\cite{AvdoninKurasov2008,BandSawickiComerford2011,ErshovaKarpenkoKiselev2016a,KiselevErshova2012,KurasovNowaczyk2005,OrenBand2011,Rueckrieme2011}.

\section{General coupling: reducibility for bipartite layers}\label{sec:asymmetric2}

When the connecting edges of a periodic bilayer quantum graph do not belong to the same asymmetry class, one can no longer expect to obtain reducibility of the Floquet surface.  Remarkably, however, it turns out that any bilayer quantum graph based on a bipartite single layer with two vertices in a fundamental domain always has reducible Floquet surface regardless of the potentials placed on the connecting edges.  A notable example of this is bilayer graphene (Fig.~\ref{fig:GrapheneW}).  Observe that the model of bilayer graphene considered here is the perfectly aligned one, known as AA-stacked.  This result does not extend to $m$-partite graphs or even bipartite when the fundamental domain has more than one vertex.

The vertices of a bipartite periodic quantum graph are divided into two sets, say green vertices and red vertices, and each edge of the graph connects vertices of two different colors.  Additionally, it will be required that any fundamental domain of the single layer graph contain exactly two vertices, one green and one red.  The set of green vertices consists of all translations of a green vertex (by the action of $\ZZ^n$), and the set of red vertices consists of all translations of a red one.

Reducibility for graphs of this kind is due to the determinant of $\hat\Afo(\lambda,z)$ being a function of a single composite Laurent polynomial of $(z_1,\dots,z_n)$. 
Not only is the Floquet surface reducible, but the Floquet modes corresponding to the two components (considered as elements of $\CC^4$ when restricted to the vertices of a fundamental domain) lie in two independent two-dimensional subspaces.  These findings are stated in Theorem~\ref{thm:graphene}.  

When the A-functions $a_v(\lambda)$ are not identical over all the vertices in a fundamental domain $\Wo$,
the form (\ref{Gf2}) of $\Gf(\lambda)$ cannot be split into two tensor products with complementary eigen-projections.  These projections depend in general on the vertex, and when there are only two vertices $v_1$ and $v_2$, one has
\begin{equation}
\begin{split}
  \Gf(\lambda) &=\, 
  s_1(\lambda)^{-1} \big[ -b_1(\lambda) + \mu_1 \left( P_{\mu_1} - P_{\!-\mu_1} \right) \big] \otimes E_1  \\
  &+\,
  s_2(\lambda)^{-1} \big[ -b_2(\lambda) + \mu_2 \left( P_{\mu_2} - P_{\!-\mu_2} \right) \big] \otimes E_2\,,
\end{split}
\end{equation}
in which $\mu_j$ for $j=1,2$ are functions on the Riemann surfaces $\left\{ (\lambda,\mu_j) : \mu_j^2 = a_j(\lambda)^2+1 \right\}$\,.

\begin{figure}
\scalebox{0.20}{\includegraphics{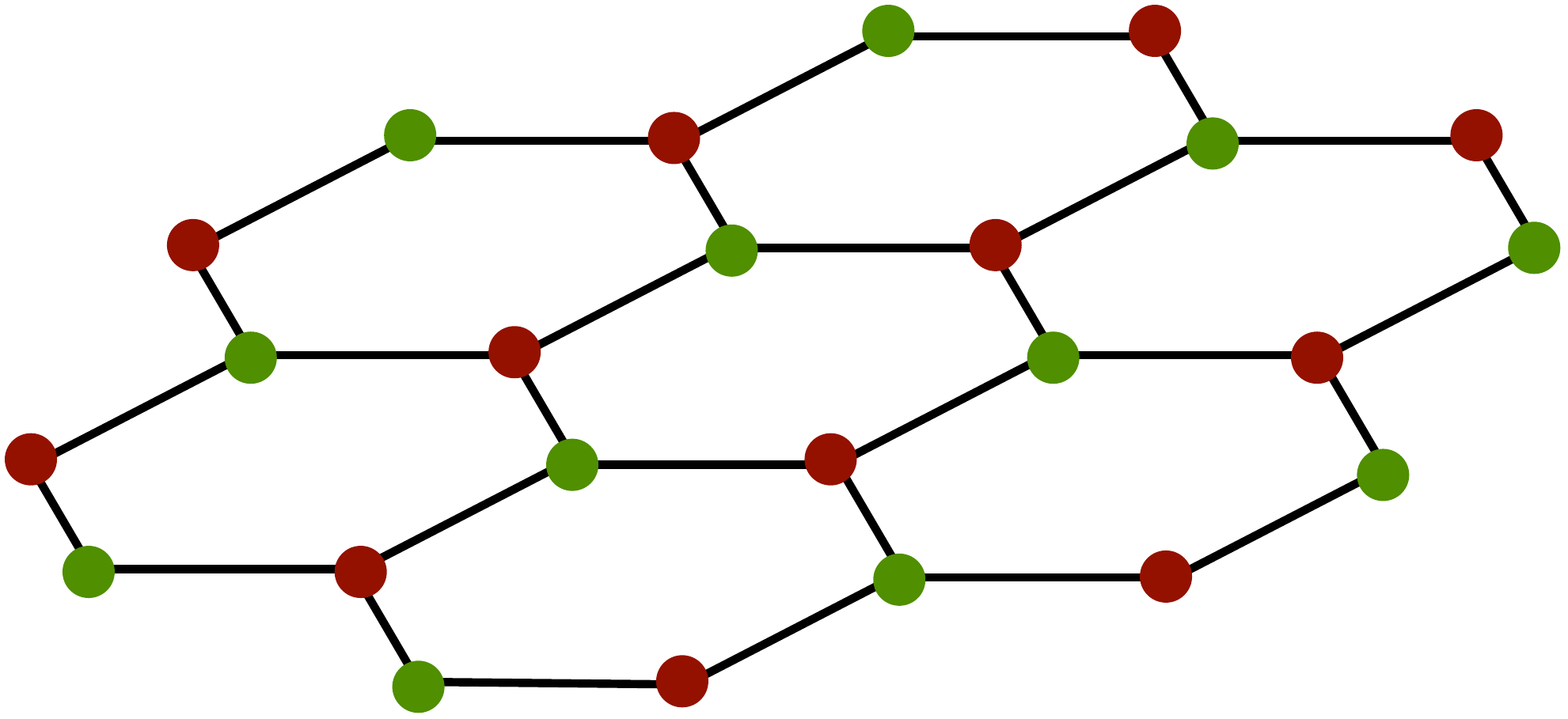}}
\hspace{0.7em}
\scalebox{0.19}{\includegraphics{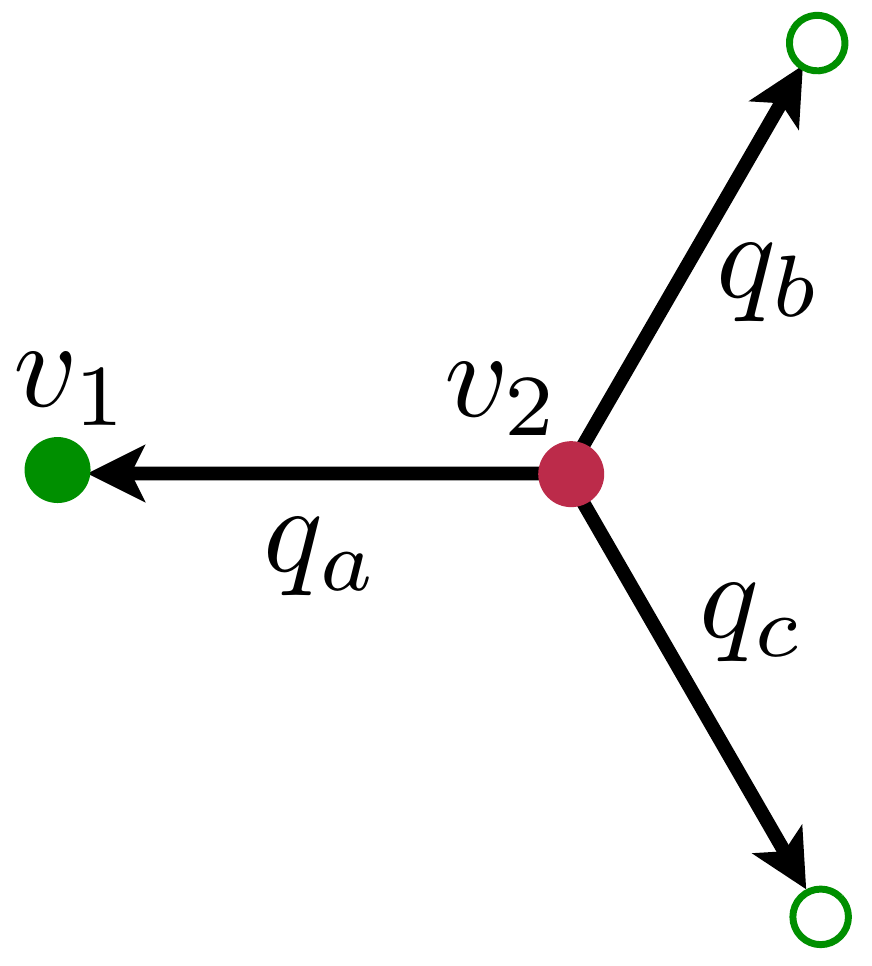}}
\hspace{3em}
\raisebox{-8pt}
{
\scalebox{0.18}{\includegraphics{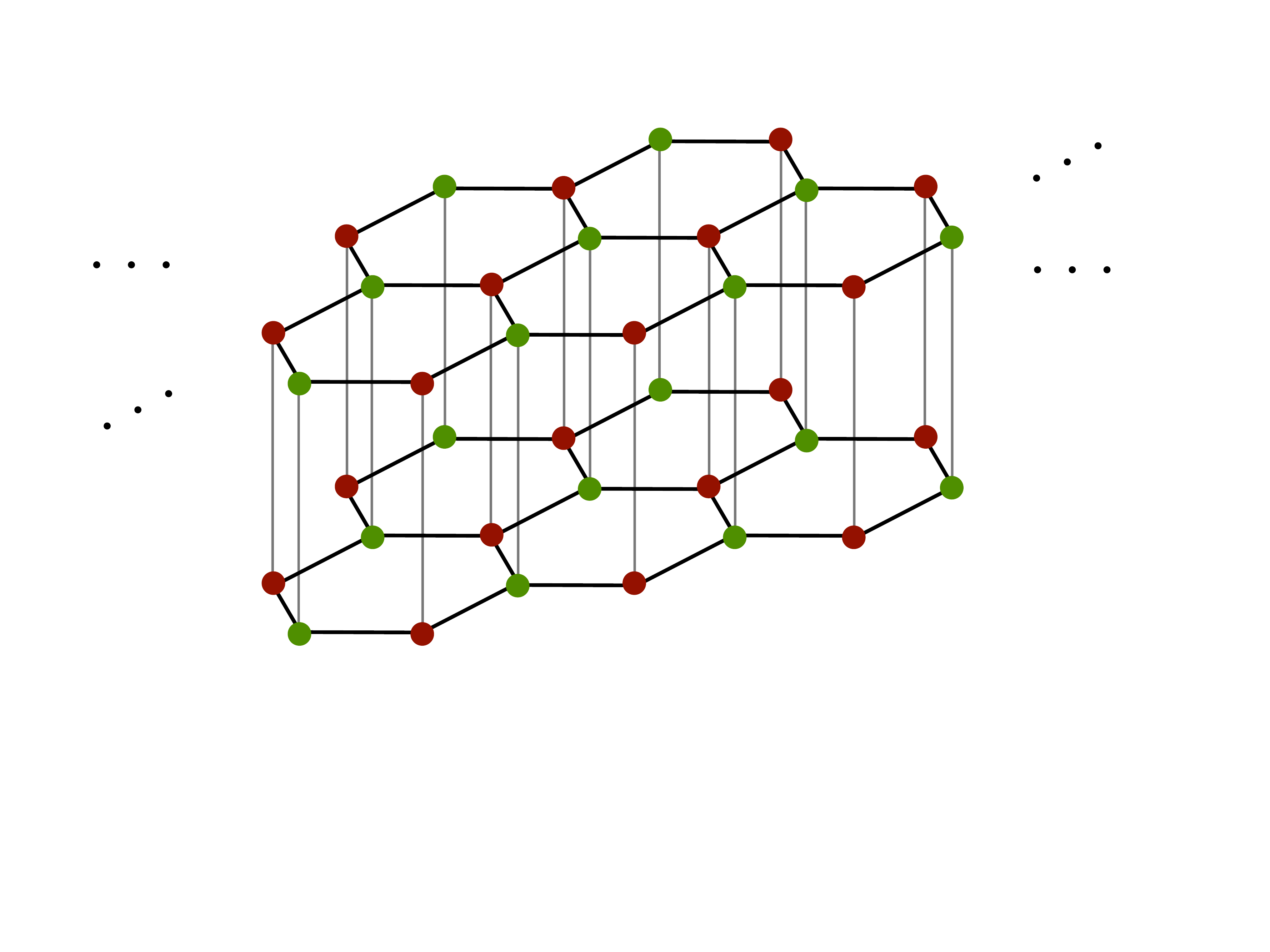}}
\hspace{1em}
\scalebox{0.17}{\includegraphics{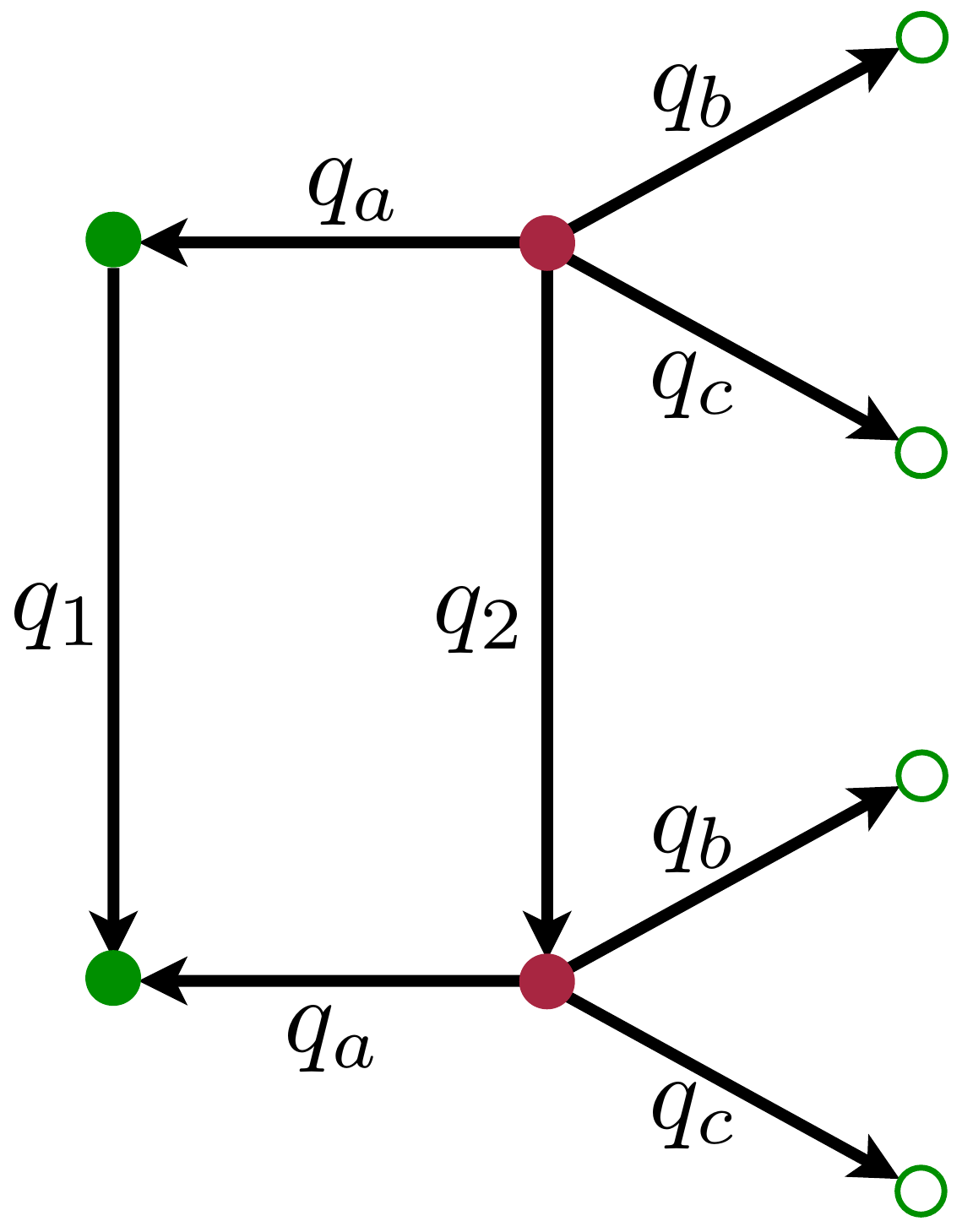}}
}
\caption{The quantum-graph model of graphene and its fundamental domain; and the associated bilayer quantum graph and its fundamental domain.}
\label{fig:GrapheneW}
\end{figure}

Let $(\Go,\Ao)$ be a bipartite periodic quantum graph with exactly one vertex of each color in each fundamental domain.  Let $(\G,A)$ be the bilayer quantum graph associated with $(\Go,\Ao)$ and any choice of potentials $q_1(x)=q_{v_1}(x)$ and $q_2(x)=q_{v_2}(x)$ on the connecting edges in a fundamental domain $W$.  Fig.~\ref{fig:GrapheneW} depicts the case of graphene.

The Floquet transform $\hat\Afo(\lambda,z)$ of the combinatorial reduction $\Afo(\lambda)$ of $\Ao$ is a $2\times2$ matrix that is meromorphic in $\lambda$ and Laurent polynomial in $z_1,\dots,z_n$.  Because the graph is bipartite, the diagonal of this matrix is independent of $z_1,\dots,z_n$, so that
\begin{equation}\label{Aohatbipartite}
  \hat\Afo(\lambda,z) \,=\,
  \renewcommand{\arraystretch}{1.3}
\left[
  \begin{array}{cc}
    m_1(\lambda) & w'(\lambda,z) \\
    w(\lambda,z) & m_2(\lambda)
  \end{array}
\right],
\end{equation}
in which $w$ and $w'$ are $\lambda$-dependent Laurent polynomials in $z_1,\dots,z_n$.

As an illustration, let $(\Go,\Ao)$ be the quantum graph model of graphene with potentials $q_a(x)$, $q_b(x)$, and $q_c(x)$ on the three edges emanating from vertex $v_1$ (see Fig.~\ref{fig:GrapheneW}) and Robin constants $\alpha_1$ and $\alpha_2$ for the two vertices $v_1$ and $v_2$.  The entries of the matrix~(\ref{Aohatbipartite}) are
\begin{eqnarray*}
  w(\lambda,z) &=& s_a(\lambda)^{-1} + s_b(\lambda)^{-1}z_1 + s_c(\lambda)^{-1}z_2\,, \\
  w'(\lambda,z) &=& s_a(\lambda)^{-1} + s_b(\lambda)^{-1}z_1^{-1} + s_c(\lambda)^{-1}z_2^{-1}\,,
\end{eqnarray*}
\begin{eqnarray*}
  m_1(\lambda) &=& -\frac{s_a'(\lambda)}{s_a(\lambda)} -\frac{s_b'(\lambda)}{s_b(\lambda)} -\frac{s_c'(\lambda)}{s_c(\lambda)} - \alpha_1, \\
  m_2(\lambda) &=& -\frac{c_a(\lambda)}{s_a(\lambda)} -\frac{c_b(\lambda)}{s_b(\lambda)} -\frac{c_c(\lambda)}{s_c(\lambda)} - \alpha_2.
\end{eqnarray*}

For each $\lambda\in\CC\setminus\sigma_D(A)$ and each $\zeta\in\CC$, define the algebraic curves
\begin{equation}
  F_\zeta(\lambda) \;:=\;
  \left\{\, z\in(\CC^*)^2 : w(\lambda,z)w'(\lambda,z) = \zeta\, \right\}.
\end{equation}

\begin{theorem}[Bilayer graphs with bipartite layers]\label{thm:graphene}  

Let $(\Go,\Ao)$ be a bipartite periodic quantum graph with exactly one vertex of each color in a fundamental domain, and let $(\G,A)$ be the bilayer graphene quantum graph constructed from $(\Go,\Ao)$ as in Definition~\ref{def:bilayer}.
Given an energy $\lambda\in\CC\setminus\sigma_D(A)$,

\begin{enumerate}

\item The Floquet surface $\Phi_\lambda$ is the zero set in $(\CC^*)^2$ of a quadratic polynomial $D_\lambda(\zeta)$ of a composite variable $\zeta = P(z_1,\dots,z_n)$, where $P(z)$ is a Laurent polynomial.
$D_\lambda(\zeta)$ is the characteristic polynomial, as a function of $\zeta$, of the matrix
\begin{equation}
  R(\lambda) \,=\, B_1(\lambda)\,B_2(\lambda)
\end{equation}
in which
\begin{equation}
B_i(\lambda) \;=\;
\renewcommand{\arraystretch}{1.2}
\left[\!
\begin{array}{cc}
  \!m_i(\lambda) - s_i(\lambda)^{-1}c_i(\lambda) & s_i(\lambda)^{-1} \\
  s_i(\lambda)^{-1} & m_i(\lambda) - s_i(\lambda)^{-1}\tilde c_i(\lambda)
\end{array}
\!\right]\,.
\end{equation}

\item  $\Phi_\lambda$ is reducible into components of the form $F_\zeta$, where $\zeta$ takes on the two eigenvalues of $R(\lambda)$.  (The components coincide when $R(\lambda)$ has an eigenvalue of multiplicity two.)

\item  Whenever $R(\lambda)$ has two distinct eigenvalues, the Floquet modes of the two components of $\Phi_\lambda$ lie in two independent subspaces of the four-dimensional space of non-$L^2$ $\lambda$-eigenfunctions of the quantum-graph operator $A$ for bilayer graphene restricted to one fundamental domain of the underlying metric graph~$\Gamma$.  If $\zeta\not=0$, the subspace for $F_\zeta$ is two-dimensional, and if $\zeta=0$ it is one-dimensional.

\end{enumerate}

\end{theorem}

\begin{proof}
The Floquet transform of the combinatorial operator $\Af(\lambda)$ corresponding to $(\G,A)$ is
\begin{equation}\label{hatAf}
  \hat\Af(\lambda,z) \,=\, I\otimes\hat\Afo(\lambda,z) + \Gf(\lambda)\,.
\end{equation}
The matrix for this operator in the basis of values of $\hat u(\lambda,z)$ at the four vertices of a fundamental domain in the order $\{(v_1,1),(v_1,2),(v_2,1),(v_2,2)\}$, is, in block form,
\begin{equation}\label{hatAf2}
  \hat\Af(\lambda,z) \,\cong\,
  \renewcommand{\arraystretch}{1.2}
\left[\!
\begin{array}{cc}
  B_1(\lambda) & w'(z,\lambda) I \\
  w(z,\lambda) I & B_2(\lambda)
\end{array}
\!\right].
\end{equation}
The determinant of this matrix is
\begin{equation}
  \det \hat\Af(\lambda,z) \;=\; \det \big( B_1(\lambda)B_2(\lambda) - ww' I\, \big),
\end{equation}
which is a quadratic polynomial of the composite variable $\zeta=ww'$ (with coefficients that are meromorphic in~$\lambda$).  This proves part (1).  That this polynomial factors into two linear factors in $\zeta$ proves part (2).

Let $\zeta$ be a nonzero root of this determinant so that $w$ and $w'$ are also nonzero, and let the nonzero vectors $\phi_1$ and $\phi_2$ in $\CC^2$ satisfy
\begin{eqnarray}
  && \left( B_1B_2 - \zeta \right)\phi_2 = 0 \\
  && \phi_1 = - B_2\phi_2 \,.
\end{eqnarray}
It is then verified that $[\phi_1,w\,\phi_2]$ is an eigenvector of the matrix in~(\ref{hatAf2}).  As $w$ runs over $\CC^*$, these vectors form the span of $[\phi_1,0]$ and $[0,\,\phi_2]$.  If $\zeta=0$, then $[\psi_1,\psi_2]$ is an eigenvector of the matrix in~(\ref{hatAf2}) where the vectors $\psi_1$ and $\psi_2$ in $\CC^2$ satisfy $B_1\psi_1=0$ and $B_2\psi_2=0$ and thus are independent of $z_1$ and $z_1$ satisfying $\zeta=0$.  This proves part (3).
\end{proof}

A more detailed description of the curves $F_\zeta$ is given by the Proposition~\ref{Fzeta} below in the case of bilayer graphene with identical potentials in the layers, that is $q_a=q_b=q_c$ (but of course $q_1$ and $q_2$ are still arbitrary).  
The following lemma aids in the proof of the proposition.  The proofs are elementary.

\begin{lemma}\label{z12}  
Numbers\, $z_1,\,z_2\in\CC^*$ and\, $q,\,q'\in\CC$ satisfy the system
\begin{equation}
  \renewcommand{\arraystretch}{1.1}
\left\{
\begin{array}{l}
  q  \,=\,  z_1 + z_2 \\
  q' =\,  z_1^{-1} + z_2^{-1}
\end{array}
\right.
\end{equation}
if and only if
\begin{equation}
\begin{split}
   \{ z_1, z_2 \} &= \left\{ z\in\CC^* : q'z + qz^{-1} = qq' \right\} \quad \text{if }\; qq'\not=0, \\
   \{ z_1, z_2 \} &\in \left\{ \{ z, -z\} : z\in\CC^* \right\} \quad \text{if }\; qq'=0.
\end{split}
\end{equation}
Additionally, $z_1=z_2$ if and only if $qq'=4$; and, in this case, $z=q/2$.
If $qq'=0$, both $q$ and $q'$ vanish.
\end{lemma}

\begin{proposition}\label{Fzeta}
\begin{eqnarray}
  F_\zeta &=& \displaystyle\bigcup_{w\in\CC\setminus\{0,1,\zeta\}} 
  \left\{ (z_1,z_2)\in(\CC^*)^2 : \left\{ z_1,z_2 \right\} = \big\{ z : \textstyle\frac{1}{w-1}z + \textstyle\frac{w}{\zeta-w}z^{-1} = 1 \big\} \right\}, \quad \text{for } \zeta\not\in\left\{ 0,1 \right\} \\
  F_1 &=& \displaystyle\bigcup_{w\in\CC\setminus\{0,1\}} 
  \left\{ (z_1,z_2)\in(\CC^*)^2 : \left\{ z_1,z_2 \right\} = \big\{ z : \textstyle\frac{1}{w-1}z + \textstyle\frac{w}{\zeta-w}z^{-1} = 1 \big\} \right\}
  \cup \left\{ (z,-z) : z\in\CC^* \right\} \\
  F_0 &=& \left\{ (z_1,z_2) : 1+z_1+z_2 = 0 \text{ or } 1+z_1^{-1}+z_2^{-1} = 0 \right\}
\end{eqnarray}
The intersection of the two surfaces whose union forms $F_0$ consists of the two points $(z_1,z_2)$ of the set $\left\{ (e^{i\pi/3},e^{-i\pi/3}), (e^{-i\pi/3},e^{i\pi/3}) \right\}$\,.
\end{proposition}

\section{Coupling with different asymmetry classes: irreducible case}\label{sec:asymmetric3}

The reducibility of the Floquet surface of bilayer graphene regardless of the asymmetry classes of the connecting edges is not typical for bilayer quantum graphs.  It is the bipartite property of graphene that is responsible.
Although irreducibility is generically expected, proving it can be involved.  This section presents a simple example in which a proof is not too long.

\begin{figure}[ht]
\centerline{
\scalebox{0.25}{\includegraphics{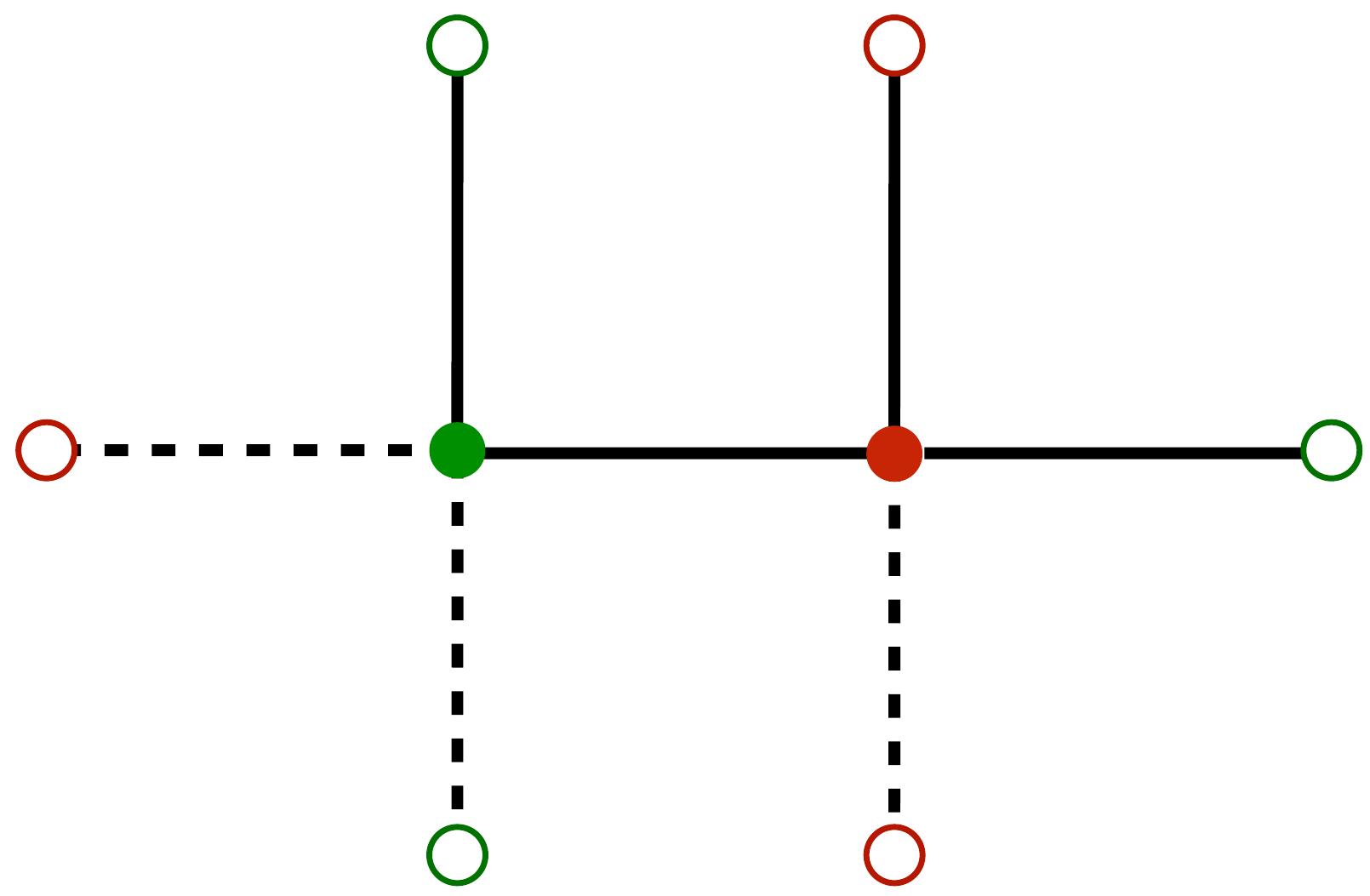}}
}
\caption{\small One fundamental domain of the single layer associated with the bilayer periodic graph constructed in section~\ref{sec:asymmetric3}.  The solid vertices and edges belong to the fundamental domain.  Vertices of the same color belong to the same $\ZZ^2$ orbit.}
\label{fig:TwoSquares}
\end{figure}

Consider a bilayer quantum graph $(\G,A)$ associated with the square periodic graph $(\Go,\Ao)$ illustrated in Fig.~\ref{fig:TwoSquares}.  A fundamental domain of $\Go$ contains two adjacent fundamental domains of a perfect square lattice, and thus contains two vertices $v_1$ and $v_2$.  In $(\G,A)$, corresponding vertices on the two copies of $\Go$ are connected by potentials with A-functions $a_1(\lambda)$ and $a_2(\lambda)$.
The Floquet transform of the combinatorial operator $\Af(\lambda)$ corresponding to $(\G,A)$ is
$
\hat\Af(\lambda,z) \,=\, I\otimes\hat\Afo(\lambda,z) + \Gf(\lambda)\,.
$
Its block form with respect to the resolution $I = P_{\mu_1}+P_{\!-\mu_1}$ by the eigen-projections for vertex $v_1$ is
\begin{equation}\label{sAhat2}
  \so(\lambda)\, \hat\Afo(\lambda,z) \,\cong\,
  \renewcommand{\arraystretch}{1.5}
\left[
  \begin{array}{cc|cc}
    \Dfcn\op (\lambda)+\zeta_2 & \xi' & 0 & 0 \\
    \xi & \Dfcn\tp (\lambda)+\zeta_2 & 0 & r(\lambda) \\
    \hline
    0 & 0 & \Dfcn\om (\lambda)+\zeta_2 & \xi' \\
    0 & r(\lambda) & \xi & \Dfcn\tm (\lambda)+\zeta_2
  \end{array}
\right]\,,
\end{equation}
in which (assuming for simplicity that the potentials of $(\Go,\Ao)$ are symmetric so that $\cto(\lambda)=\co(\lambda)$)
\begin{equation}
\begin{split}
  \Dfcn_j^\pm(\lambda) &=\, -4\co(\lambda) - \alpha\so(\lambda) \,+\, \so(\lambda) s_j(\lambda)^{-1}\!\left( -b_j(\lambda) \pm \nu_j \right),     \quad \text{for } j=1,2 \\
  r(\lambda) &=\, 0 \quad \text{if and only if } a_1(\lambda) = a_2(\lambda)\,,
\end{split}
\end{equation}
and and $\nu_1=\mu_1$.
The dependence on $z_1$ and $z_2$ comes through
\begin{equation}
\renewcommand{\arraystretch}{1.2}
\left.
\begin{array}{ll}
  \xi = 1+z_1\,, & \xi' = 1+z_1^{-1}, \\
  \zeta_1 = z_1 + z_1^{-1}, & \zeta_2 = z_2 + z_2^{-1}, \\
  w_1w_1' = 2 + \zeta_1\,.
\end{array}
\right.\end{equation}

\begin{proposition}
  The Floquet surface of the periodic bilayer quantum graph $(\G,A)$ described above is reducible at energy $\lambda$ if and only if 
\begin{equation}
  r^2(\lambda) \left( \Dfcn\op (\lambda) - \Dfcn\om  (\lambda)\right)^2 \left( r^2(\lambda) - \left( \Dfcn\op (\lambda) - \Dfcn\tm (\lambda) \right)\left( \Dfcn\tp (\lambda) - \Dfcn\om (\lambda) \right) \right) \,=\, 0\,.
\end{equation}
\end{proposition}

As expected, this meromorphic function of $\lambda$ vanishes identically when $r(\lambda)\equiv0$, which occurs exactly when the A-functions $a_1(\lambda)$ and $a_2(\lambda)$ for the potentials of the two edges are identical.

\begin{proof}
The determinant of (\ref{sAhat2}) turns out to be a function of $\zeta_1$ and $\zeta_2$, which means that it is invariant under $z_1\mapsto z_1^{-1}$ and $z_2\mapsto z_2^{-1}$,
\begin{equation}
  D(\lambda;z_1,z_2) \,=\,
  \zeta_1^2 + \zeta_1 \left( p^+(\lambda,\zeta_2) + p^-(\lambda,\zeta_2) \right)
                   + p^+(\lambda,\zeta_2) p^-(\lambda,\zeta_2) - r(\lambda)^2 q(\lambda,\zeta_2)\,,
\end{equation}
in which
\begin{eqnarray}
  p^+(\zeta_2) &=& 2 - (\Dfcn\op +\zeta_2)(\Dfcn\tp +\zeta_2) \\
  p^-(\zeta_2) &=& 2 - (\Dfcn\om +\zeta_2)(\Dfcn\tm +\zeta_2) \\
  q(\zeta_2) &=& \left( \Dfcn\op  + \zeta_2 \right) \left( \Dfcn\om  + \zeta_2 \right)\,.
\end{eqnarray}
Suppress the dependence on $\lambda$ and assume that $D$ has the factorization
\begin{equation}
\begin{split}
  D(z_1,z_2) &=
  \left( z_1 + \gamma(z_2) + \beta z_1^{-1} \right)
  \left( z_1 + \delta(z_2) + \beta^{-1} z_1^{-1} \right), \\
  &= z_1^2+z_1^{-2} + z_1(\gamma(z_2)+\delta(z_2)) + z_1^{-1} (\beta^{-1}\gamma(z_2) + \beta\delta(z_2))
        + \beta + \beta^{-1} + \gamma(z_2)\delta(z_2)
\end{split}
\end{equation}
in which $\gamma(z_2)$ and $\delta(z_2)$ are Laurent polynomials and $\beta$, as well as the coefficients of $\gamma(z_2)$ and $\delta(z_2)$ depend on~$\lambda$.  This leads to the equations
\begin{equation}
\begin{split}
  \gamma(z_2) + \delta(z_2) = \beta^{-1}\gamma(z_2) + \beta\delta(z_2) &= p^+(z_2) + p^-(z_2) \\
  \beta + \beta^{-1} + \gamma(z_2)\delta(z_2) &= 2 + p^+(z_2)p^-(z_2) - r^2 q(z_2)\,.
\end{split}
\end{equation}
Because of this, $\gamma(z_2)$ and $\delta(z_2)$ must have the form
\begin{equation}
\begin{split}
  -\gamma(z_2) &\,=\, z_2^2 + \gamma_1z_2 + \gamma_0 + \gamma_{-1}z_2^{-1} + z_2^{-2} \\
  -\delta(z_2) &\,=\, z_2^2 + \delta_1z_2 + \delta_0 + \delta_{-1}z_2^{-1} + z_2^{-2} \,.
\end{split}
\end{equation}

Now it is shown that $\gamma(z)$ and $\delta(z)$ (putting $z_2=z$ for now) are functions of $\zeta=z+z^{-1}$, which is to say that $\gamma_1=\gamma_{-1}$ and $\delta_1=\delta_{-1}$.  They can be written in terms of $\zeta$ and $\zeta^-=z-z^{-1}$,
\begin{equation}
\begin{split}
  -\gamma(z) &\,=\, \zeta^2 + \gamma_+\zeta + \gamma_-\zeta^- + \tilde\gamma_0 \\
  -\delta(z) &\,=\, \zeta^2 + \delta_+\zeta + \delta_-\zeta^- + \tilde\delta_0\,,
\end{split}
\end{equation}
in which $\tilde\gamma_0=\gamma_0-2$ and $\tilde\delta_0=\delta_0-2$.
Their product is
\begin{equation}
\begin{split}
  \gamma(z)\delta(z) &\;=\;
  \zeta^4 + \zeta^3(\gamma_++\delta_+) + \zeta^2(\gamma_+\delta_+ + \gamma_-\delta_- +\tilde\gamma_0+\tilde\delta_0)
  + \zeta(\gamma_+\tilde\delta_0+\delta_+\tilde\gamma_0) + \tilde\gamma_0\tilde\delta_0 - 4\gamma_-\delta_- \\
  &\;+\; \zeta^- \left( \gamma_-(\zeta^2+\delta_+\zeta+\tilde\delta_0) + \delta_-(\zeta^2+\gamma_+\zeta+\tilde\gamma_0) \right),
\end{split}
\end{equation}
and it is a function of $\zeta$.  Thus the latter expression multiplying $\zeta^-$ must vanish identically.  This is equivalent to the vanishing of three quantities:
\begin{eqnarray}
  && \gamma_- + \delta_- \;=\; 0\,, \\
  && \gamma_- \delta_+ + \delta_-\gamma_+ \;=\; 0\,, \\
  && \gamma_-\delta_0 + \delta_-\gamma_0 \:=\: 0\,.
\end{eqnarray}
As long as not all of $\delta_+$, $\gamma_+$, $\delta_0$, and $\gamma_0$ are zero, one obtains $\gamma_-=0$ and $\delta_-=0$.  This makes $\gamma(z)$ and $\delta(z)$ polynomial functions of $\zeta$.

This yields the factorization
\begin{equation}
\begin{split}
  D(z_1,z_2)  &\;=\; \zeta_2^2 + \zeta_1\left( p^+(\zeta_2) + p^-(\zeta_2) \right)
                               + p^+(\zeta_2) p^-(\zeta_2) - r^2 q(\zeta_2) \\
                      &\;=\; \big( \zeta_1 + \gamma(\zeta_2) \big) \big( \zeta_1 + \delta(\zeta_2) \big).
\end{split}
\end{equation}
So the roots of $D$ as a function of $\zeta_1$ are polynomials in $\zeta_2$.  This means that the $\zeta_1$-discriminant
\begin{equation}
  D_1(\zeta_2) \;=\; \big( p^+(\zeta_2) - p^-(\zeta_2) \big)^2 \,+\, {\color{black}4\,} r^2\, q(\zeta_2)
\end{equation}
is the square of a polynomial.  This, in turn, means that the $\zeta_2$-discriminant
\begin{equation}
  D_2 \;=\; \left[ 2cd + 4r^2(\Dfcn\op  + \Dfcn\om ) \right]^2 - 4\left( c^2 - 4r^2 \right) \left( d^2 - 4r^2\Dfcn\op  \Dfcn\om  \right)
\end{equation}
vanishes.  A page of calculations yields
\begin{equation}
  D_2 \;=\;  16\,r^2 \left( \Dfcn\op  - \Dfcn\om  \right)^2 \left( r^2 - \left( \Dfcn\op  - \Dfcn\tm  \right)\left( \Dfcn\tp  - \Dfcn\om  \right) \right)\,.
\end{equation}
Recall now that all the quantities here are meromorphic functions of $\lambda$.

The other possible factorization of $D$ is
\begin{equation}
  D(z_1,z_2) \;=\; \left( 1 + \beta z_1^{-1} \right)
          \left( z_1^2 + \gamma(z_2) z_1 + \delta(z_2) + \beta^{-1} z_1^{-1} \right),
\end{equation}
which leads to the equations
\begin{equation}
\begin{split}
  \gamma(z_2) + \beta = \beta\delta(z_2) + \beta^{-1} &= p^+(z_2) + p^-(z_2) \\
  \delta(z_2) + \beta\gamma(z_2) &= 2 + p^+(z_2)p^-(z_2) - r^2 q(z_2)\,.
\end{split}
\end{equation}
The first of these implies that $\gamma$ and $\delta$ are both of the form
$-z^2 + \dots - z^{-2}$, which is untenable in view of the second of these equations.
\end{proof}

\bigskip
\bigskip

\noindent
{\bfseries\large Acknowledgment.}  
This research was supported by NSF Research Grant DMS-1411393.

\end{document}